\documentclass[11pt]{amsart}

\usepackage{bbold}
\usepackage{amscd} 
\usepackage{amsmath,amssymb}
\usepackage{mathtools}
\usepackage{mathrsfs}
\usepackage{tikz}
\newtheorem{theorem}{Theorem}[section]

\theoremstyle{definition}

\newtheorem{remark}[theorem]{Remark} 
\numberwithin{equation}{section}

\def\RE{\mathbb R}
\def\R{\mathbb R}
\def\CO{\mathbb C}
\def\C{\mathbb C}
\def\H{\mathsf H}
\def\D{\mathsf D}

\def\oneb{\mathbb 1}
\def\uno{\mathbb{1}}
\def\zero{\mathbb{0}}

\def\sgn{\text{sgn}}

\def\ran{\text{\rm ran}}
\def\ker{\text{\rm ker}}
\def\dom{\text{\rm dom}}
\def\be{\begin{equation}}
\def\ee{\end{equation}}

\title[D${\,}^{\!2}$=H+1/4 with point interactions]
{ D${\,}^{\!2}$=H+$\frac14$ with point interactions}

\begin{document}
\author{Andrea Posilicano}
\author{Linda Reginato}
\address{DiSAT, Sezione di Matematica, Universit\`a dell'Insubria, via Valleggio 11, I-22100
Como, Italy}
\email{andrea.posilicano@unisubria.it}
\email{linda.reginato98@gmail.com}
\begin{abstract} Let ${\mathsf D}$ and ${\mathsf H}$ be the self-adjoint, one-dimensional Dirac and Schr\"odinger operators in $L^{2}(\mathbb{R};\mathbb{C}^{2})$ and $L^{2}(\mathbb{R};\mathbb{C})$ respectively. It is well known that, in absence of an external potential, the two operators are related through the equality ${\mathsf D}^2 = ({\mathsf H} + \frac{1}{4}){\mathbb 1}$. We show that such a kind of relation also holds in the case of $n$-point singular perturbations: given any  self-adjoint realization $\widehat {\mathsf D}$ of the formal sum ${\mathsf D}+\sum_{k=1}^{n}\gamma_{k}\delta_{y_{k}}$, we explicitly determine the self-adjoint realization $\widehat{\mathsf H}$ of 
${\mathsf H}{\mathbb 1}+\sum_{k=1}^{n}(\alpha_{k}\delta_{y_{k}}+\beta_{k}\delta'_{y_{k}})$ such that ${\widehat{\mathsf D}}^2 = \widehat{\mathsf H} + \frac{{\mathbb 1}}{4}$. The found correspondence preserves the subclasses of self-adjoint realizations corresponding to both the local and the separating boundary conditions. Some connections with supersymmetry are provided. The case of nonlocal boundary conditions allows the study of the relation ${\mathsf D}^{2}={\mathsf H}+\frac14$  for quantum graphs with (at most) two ends; in particular, the square of the extension corresponding to Kirchhoff-type boundary conditions for the Dirac operator on the graph gives the direct sum of two Schr\"odinger operators on the same graph, one with the usual Kirchhoff boundary conditions and the other with a sort of reversed Kirchhoff ones.
\end{abstract}
\maketitle
\section{Introduction}\label{S1}
Let $L^2(\R;\C^d)$ be the Hilbert space of $\C^{d}$-valued square integrable functions with scalar product $\langle\cdot, \cdot\rangle$ and norm $\|\cdot \|$; likewise, $H^2(\mathbb{R};\mathbb{C}^{d})\subset H^1(\mathbb{R};\mathbb{C}^{d})\subset C_{b}(\mathbb{R};\mathbb{C}^{d})$ denote the Sobolev space on $\R$ of order 1 and 2 and the space of bounded continuous functions  with values in $\C^d$ respectively. Whenever $d=1$, we simply write $L^2(\R)$, $H^k(\mathbb{R})$ and $ C_{b}(\mathbb{R})$. 
In $L^{2}(\RE;\CO^{2})$	we consider the free self-adjoint Dirac operator $\D$ defined by
\[
\D: H^1(\mathbb{R};\mathbb{C}^{2})\subseteq L^2(\mathbb{R};\mathbb{C}^{2})\to L^2(\mathbb{R};\mathbb{C}^{2})\,,\quad 
\D := -i\, \frac{d\,}{d x} \,\sigma_1 + \frac{1}{2}\, \sigma_3\,,
\]
where  $\sigma_{1}$ and $\sigma _{3}$ are the Pauli matrices
\[
		\sigma_ 1=
		\begin{bmatrix}
			0 & 1\\
			1 & 0\\
		\end{bmatrix},
	\qquad 
		\sigma_ 3=
		\begin{bmatrix}
			1 & 0\\
			0 & -1\\
		\end{bmatrix}.
\]
Furthermore,  we consider the free self-adjoint Schr\"odinger operator in $L^{2}(\RE)$
\[
\H :H^2(\mathbb{R}) \subseteq L^2(\mathbb{R})\to L^2(\mathbb{R}), \quad \H :=-\frac{d^{2}}{d x^{2}}\,.
\]
It is well known and easy to check that in this free case there exists a relation between the two operators:
\begin{align}\label{Rel}
		\D^2 = \left(\H + \frac{1}{4}\right)\!\oneb\,.
\end{align}
Here and below, we use the isomorphism $L^{2}(\RE;\CO^{2})\simeq L^{2}(\RE)\oplus L^{2}(\RE)$ and the identification ${\mathsf L}\oneb\equiv{\mathsf L}\oplus{\mathsf L}$, ${\mathsf L}$ a linear operator in $L^{2}(\RE)$. More generally, in the following we use the shorthand notation ${\mathsf L}\oneb\equiv{\mathsf L}\oplus{\mathsf L}$ for a linear operator $L:\dom(L)\subseteq H_{1}\to H_{2}$.
\par
Notice that \eqref{Rel} entails a relation between the resolvent operators:
\be\label{res}
(-\D+z)^{-1}=(\D+z)\left(-\H+z^{2}-\frac14\right)^{\!\!-1}\!\!\!\oneb\,,\qquad z\in \C \backslash ((-\infty, -{1}/{2}] \cup [{1}/{2}, +\infty))\,.
\ee	
The aim of this paper is to extend this connection between Dirac's and Schr\"odinger's operators to the case where $\D$ is perturbed by a sum of $\delta$'s potential, equivalently, given any self-adjoint extension $\D_{\Pi,\Theta}$ of the symmetric operator $\D|C^{\infty}_{comp}(\RE\backslash  \{y_{1},\dots,y_{n}\};\CO^{2})$, we explicitly determine  the couple $(\widehat\Pi,\widehat\Theta)$ such that    
\be\label{intro1}
(\D_{\Pi,\Theta})^{2}=\left(\widehat\H_{\widehat\Pi,\widehat\Theta}+\frac{\uno}4\right)\,.
\ee
Here, we parametrize the self-adjoint extensions of  $\D|C^{\infty}_{comp}(\RE\backslash  \{y_{1},\dots,y_{n}\};\CO^{2})$ by couples $(\Pi,\Theta)$, $\Pi:\CO^{2n}\to\CO^{2n}$ an orthogonal projector, $\Theta:\ran(\Pi)\to\ran(\Pi)$ a symmetric operator, and likewise $\widehat\H_{\widehat\Pi,\widehat\Theta}$ denotes the self-adjoint extension of $\H\uno|C^{\infty}_{comp}(\RE\backslash  \{y_{1},\dots,y_{n}\};\CO^{2})$ corresponding to the couple $(\widehat\Pi,\widehat\Theta)$,  $\widehat\Pi:\CO^{4n}\to\widehat\CO^{4n}$ an orthogonal projector,  $\widehat\Theta:\ran(\widehat\Pi)\to\ran(\widehat\Pi)$ a symmetric operator. Any operator of the kind $\widehat\H_{\widehat\Pi,\widehat\Theta}$ is a self-adjoint realization of  a singular perturbation of $\H\uno$ by a sum of $\delta$'s and $\delta'$'s potentials. As in the free case, the relation \eqref{intro1} entails another one for the resolvents:
$$
(-\D_{\Pi,\Theta}+z)^{-1}=(\D_{\Pi,\Theta}+z)\left(-\widehat\H_{\widehat\Pi,\widehat\Theta}+z^{2}-\frac{\uno}4\right)^{\!\!-1}\,,
$$ 
where $\pm z\in\varrho(\D_{\Pi,\Theta})$ if and only if $(z^{2}-\frac14)\in\varrho(\widehat\H_{\widehat\Pi,\widehat\Theta})$; here, $\varrho(L)$ denotes the resolvent set of the closed operator $L$. \par 
The specific case here considered is an example of solution of the problem concerning the representation of the square of a singular perturbation of a self-adjoint operator $A$  by a singular perturbation of $A^{2}$. This problem has been studied in \cite{AKK}; however, in such a paper only the case $A>0$ has been considered and the explicit examples there presented are limited to rank-one singular perturbations. The methods here used are different from the ones in \cite{AKK}, we do not use the resolvent  formulae directly but instead use the self-adjointness domains.\par
In more detail, the content of the paper  is the following. In Section \ref{S2} we build the whole families of the self-adjoint extensions of $\D|C^{\infty}_{comp}(\RE\backslash  \{y_{1},\dots,y_{n}\};\CO^{2})$ and $\H\uno|C^{\infty}_{comp}(\RE\backslash  \{y_{1},\dots,y_{n}\};\CO^{2})$. Instead of using the standard von Neumann theory (see, e.g., \cite{AGHKH}, \cite{BD}, \cite{DG}), which gives a parametrization in terms of unitary operators between the defect spaces, we found more convenient to use the equivalent approach proposed in \cite{JFA} and \cite{O&M}, which gives a parametrization in terms of couples $(\Pi,\Theta)$, where $\Pi$ is an orthogonal projection and $\Theta$ is a self-adjoint operator  in $\ran(\Pi)$; this allows for an easy writing of the corresponding resolvents. Then, in Section \ref{S3}, by a comparison of the self-adjointness domains,  we found the correspondence between the couple $(\Pi,\Theta)$ and $(\widehat\Pi,\widehat\Theta)$ such that \eqref{intro1} holds. In order to enhance the reader intuition, we start with  simplest case, where  $n = 1$ and $\Pi = \uno$ and then proceed step-by-step towards the most general case. Finally, in Section \ref{S4}, we present various applications. In Subsection \ref{LBC} we consider the subclass of self-adjoint extensions for the Dirac operator corresponding to local boundary conditions, i.e., to the ones which do not couple different points $y_{k}$ and show that the corresponding extensions for the Schr\"odinger operator provide local boundary conditions as well. As a particular case of  such a  result, in Subsection \ref{GS} we consider the Gesztesy-\v{S}eba realizations; they are the self-adjoint realizations of the Dirac operator with local point interactions corresponding, in the non relativistic limit, to Schr\"odinger operators with local point interactions either of $\delta$-type or of $\delta'$-type (see \cite{GS}, \cite[Appendix J]{AGHKH}, \cite{CMP}). Then, in Subsection 
\ref{SBC}, we consider the subclass of self-adjoint extensions for the Dirac operator corresponding to separating (a.k.a. decoupling) boundary conditions, i.e., to the local ones for which, at any point, left limits are independent from right limits. This entails that the corresponding Dirac operator is the direct sum of self-adjoint Dirac operators $\D_{k}$ in $L^{2}(I_{k})$, where the $I_{k}$'s are either the half-lines $(-\infty, y_{1})$ and $(y_{n},+\infty)$ or the bounded intervals $(y_{k}, y_{k+1})$; the same is true for the corresponding corresponding Schr\"odinger operator and $(\D_{k})^{2}=\widehat\H_{k}+\frac{\uno}4$. In Subsection \ref{SUSY}, some connections with supersymmetry are discussed and a simple criterion of spontaneous supersymmetry breaking is provided (see \cite{UT}, \cite{ASSW} and references therein for somehow different aspects of supersymmetry in presence of point interactions). In Subsection \ref{QG}, we point out that our results, in the case of non local boundary conditions, allow the study of the connection between the square of the Dirac operator and the Schr\"odinger operator on quantum graphs with (at most) two ends. In particular, as an explicit example, we consider the Dirac operator on the eye graph with Kirchhoff-type boundary conditions at the vertices and show that its square is the direct sum of two Schr\"odinger operators on the same graph, one with Kirchhoff boundary conditions and the other with a sort of inverse Kirchhoff ones. These latter boundary conditions, like the Kirchhoff ones, reduce, in the case of the real line, to the free boundary conditions; this is consistent with \eqref{Rel}. The procedure used for the eye graph can be extended, without substantial changes, to any kind of graph, thus showing that  the property of conservation of Kirchhoff-like boundary conditions holds in general. 
\par 
We presume that the results  here presented can be extended to the more  involved cases corresponding 
to extensions of symmetric operators with infinite deficiency indices as the 1-dimensional Dirac and Schr\"odinger operators with singular perturbations on discrete sets (see \cite{KM} and \cite{CMP}) and the $n$-dimensional ($n=2,3$) Dirac and Schr\"odinger operators with singular perturbations on 1-codimensional surfaces (see, e.g.,  \cite{ArMaVe 1}, \cite{BEHL19}, \cite{CLMT} and \cite{BLL}, \cite{JDE16}). 

\section{$\D$ and $\H$ with point interactions}\label{S2}
Given a finite set of points $Y = \{y_1, \cdots, y_n\}$, $y_{1}<y_{2},\dots<y_{n}$, we define 
\be\label{Hk}
	H^1(\RE\backslash Y;\CO^{d}) := H^1(I_{0};\CO^{d})\oplus\cdots\oplus H^1(I_{n};\CO^{d})\,,
	\ee
	where, 
	\be\label{Ik}
	I_{0}:=(-\infty,y_1)\,,\quad I_{1}:=(y_{1}, y_{2})\,,\quad  \dots\dots\quad  I_{n-1}:=(y_{n-1},y_{n})\,,\quad I_{n}:=(y_{n},+\infty)\,,
	\ee
	and
	\begin{equation*}
			H^1(I_{j};\CO^{d}) := \{ f \in L^2(I_{j};\CO^{d}): f^\prime \in L^2(I_{j};\CO^{d})\}\,,\quad j=0,\dots,n\,.
		\end{equation*}
Here and below, $f'$ denotes the (distributional) derivative of $f$. Notice that the left and right limits $f(y_{k}^{\pm})$ exists and are finite for any $f\in H^1(\RE\backslash Y;\CO^{d})$.
	We define 
	$$
	H^2(\RE\backslash Y;\CO^{d}) := H^2(I_{0};\CO^{d})\oplus \cdots\oplus H^2(I_{n};\CO^{d})\,,
	$$ 
where 
	\begin{equation*}
			H^2(I_{j};\CO^{d}) := \{ f \in H^1(I_{j};\CO^{d}): f''\in L^2(I_{j};\CO^{d})\}\,,\quad j=0,\dots,n\,.
		\end{equation*}
		Obviously, $$H^2(\RE\backslash Y;\CO^{d})\subset H^1(\RE\backslash Y;\CO^{d})\subset L^{2}(\RE;\CO^{d})$$ and $f\in H^2(\RE\backslash Y;\CO^{d})$ implies $f'\in H^1(\RE\backslash Y;\CO^{d})$. We simply write $H^k(\RE\backslash Y)$, $k=1,2$, whenever $d=1$. Next, we introduce the two bounded operators 
		\begin{align}
			\tau:H^1(\RE\backslash Y;\CO^{2})\to\CO^{2n}\,,\quad \tau \Psi:=(\tau_{y_{1}}\Psi\,,\dots,\tau_{y_{n}}\Psi)\,,\qquad \tau_{y}\Psi:=\langle \Psi\rangle_{y}\,,
			\label{tau}
		\end{align}
	and	
	\begin{align}
		\widehat\tau:H^2(\RE\backslash Y)\to\CO^{2n}\,,\qquad \widehat\tau \psi:=
		(\widehat\tau_{y_{1}}\psi\,,\dots,\widehat\tau_{y_{n}}\psi)\,,\qquad \widehat\tau_{y}\psi:=\langle \psi\rangle_{y}\oplus\langle \psi'\rangle_{y}\,,
		\label{hattau}
	\end{align}
		where 
	$$	\langle f\rangle_{y}:=\frac12\,\big(f(y^{-})+f(y^{+})\big)\,.
	$$
Clearly, $\langle f\rangle_{y_{k}}=f(y_{k})$ whenever $f\in H^{1}(\RE;\CO^{d})\subset C_{b}(\RE;\CO^{d})$.\par
In this section, following the scheme proposed in \cite{O&M} (for the equivalent approaches which use either von Neuman's theory or Boundary Triples theory, see, e.g.,  \cite{DG}, \cite{BD} and \cite[Sect. 4.1]{Pank}, \cite{KM}, \cite{CMP} respectively), we review the construction of the self-adjoint extensions of the closed symmetric operators 
$$S:=\D{|\ker(\tau|H^{1}(\RE;\CO^{2}))}\,,\qquad \widehat S:=\H{|\ker(\widehat\tau|H^{2}(\RE))}\,.
$$
Both $S$ and $\widehat S$ have defect indices $(2n,2n)$; they are the closures of the symmetric operators 
$$S^{\circ}:=\D{|C^{\infty}_{comp}(\RE\backslash Y;\CO^{2})}\,,\qquad \widehat S^{\circ}:=\H{|C^{\infty}_{comp}(\RE\backslash Y)}\,.
$$
Let $\widehat g_{z}(x-y)$ be the kernel of the free Schr\"odinger resolvent $(-\H+ z)^{-1}=\left(\frac{d^{2}}{d x^{2}}+ z\right)^{-1}$,  with $z\in \varrho(\H)=\CO\backslash[0,+\infty)$, i.e.,
\begin{align}
	\widehat g_{z}(x)= \frac{e^{i\sqrt{z}\,|x|}}{2i\sqrt z}\,,\qquad \text{Im}(\sqrt z)>0\,.
	\label{hatg}
\end{align}
By \eqref{res}, setting $w_{z}:=z^{2}-\frac14$, one then obtains the kernel $g_{z}(x-y)$ of the free Dirac resolvent $(-\D+z)^{-1}$, $z\in\varrho(\D)=\C \backslash ((-\infty, -{1}/{2}] \cup [{1}/{2}, +\infty))$,
\begin{align}
	g_{z}(x)=(\D+z)\widehat g_{w_{z}}\oneb=\frac{e^{i\sqrt{w_{z}}\,|x|}}{2i} \begin{bmatrix}	\zeta_{z} & \sgn(x)\\	\sgn(x) & \zeta_{z}^{-1}\\\end{bmatrix},
	\label{g}
\end{align}
	where $\zeta_{z} := ({\frac{1}{2}-z})/{\sqrt{w_{z}}}$ and $\text{Im}(w_{z})>0$. 
By such kernels, one gets that the bounded operators 
$$
G_{z}:\CO^{2n}\to L^{2}(\RE;\CO^{2})\,,\quad G_{z}:=(\tau(-\D+\bar z)^{-1})^{*}\,,\quad z\in
\C \backslash ((-\infty, -{1}/{2}] \cup [{1}/{2}, +\infty))\,,
$$	
and 
$$
\widehat G_{z}:\CO^{2n}\to L^{2}(\RE)\,,\quad \widehat G_{z}:=(\widehat\tau(-\H+\bar z)^{-1})^{*}\,,\quad z\in\CO\backslash[0,+\infty)\,,
$$	
represents as 
$$
[G_{z}\xi](x)=\sum_{k=1}^{n}g_{z}(y_{k}-x)\,\xi_{k}\,,\qquad \xi\equiv(\xi_{1},\dots,\xi_{n})\,,\quad \xi_{k}\in\CO^{2}\,.
$$
and 
$$
[\widehat G_{z}\xi](x)=\sum_{k=1}^{n}(\widehat g_{z}(y_{k}-x)\,\xi_{k,1} +\widehat g^{\,\prime}_{z}(y_{k}-x)\,\xi_{k,2})\,,\quad \xi\equiv((\xi_{1,1},\xi_{1,2}),\dots,(\xi_{n,1},\xi_{n,2}))\,.
$$
Their adjoints 
$$
G^{*}_{\bar z}:L^{2}(\RE;\CO^{2})\to \CO^{2n} \,,\qquad
\widehat G^{*}_{\bar z}: L^{2}(\RE)\to\CO^{2n}
$$	
are given by  
$$
G^{*}_{\bar z}\Psi=\big((G^{*}_{z}\Psi)_{1},\dots,(G^{*}_{z}\Psi)_{n}\big)\,,\qquad 
(G^{*}_{\bar z}\Psi)_{k}:=\int_{\RE}g_{z}(y_{k}-x)\Psi(x)\,dx
$$
and 
$$
\widehat G^{*}_{\bar z}\psi=\big((\widehat G^{*}_{z}\Psi)_{1},\dots,(\widehat G^{*}_{z}\Psi)_{n}\big)\,,\qquad 
(\widehat G^{*}_{\bar z}\psi)_{k}:=\left(\int_{\RE}\widehat g_{z}(y_{k}-x)\psi(x)\,dx,\int_{\RE}\widehat g^{\,\prime}_{z}(y_{k}-x)\psi(x)\,dx\right).
$$
Since 
$$
\text{$G_{z}\xi\in H^{1}(\RE\backslash Y;\CO^{2})\quad$ and $\quad\widehat G_{z}\xi\in H^{2}(\RE\backslash Y)$,}
$$ both 
$$
\text{$\tau G_{z}:\CO^{2n}\to \CO^{2n}\quad$ and $\quad\widehat\tau \widehat G_{z}:\CO^{2n}\to \CO^{2n}$}
$$ 
are well defined and are represented by the two $n\times n$ block matrices  with the $2\times 2$ blocks
\be\label{Djk}
[\tau G_{z}]_{jk}=\frac{e^{i\sqrt{w_{z}}\,|y_{k}-y_{j}|}}{2i} \begin{bmatrix}	\zeta_{z} & \sgn(y_{k}-y_{j})\\	\sgn(y_{k}-y_{j}) & \zeta_{z}^{-1}\\\end{bmatrix},
\ee
\be\label{Hjk}
[\widehat\tau \widehat G_{z}]_{jk}=\frac{e^{i\sqrt{z}\,|y_{k}-y_{j}|}}{2} \begin{bmatrix}	(i\sqrt z)^{-1} & \sgn(y_{k}-y_{j})\\	-\sgn(y_{k}-y_{j}) & i\sqrt{z}\end{bmatrix},
\ee
where
$$
\sgn(x):=\begin{cases}-1&x<0\\0&x=0\\+1&x>0\,.\end{cases}
$$
In the following, given an orthogonal projection $P:\CO^{d}\to\CO^{d}$, 
by a slight abuse of notation, we use the same symbol to denote both the surjection $P:\CO^{d}\to\ran(P)$ and the injection $P:\ran(P)\to\CO^{d}$.
\begin{theorem}\label{K1}
The sets of self-adjoint extensions of $S$ and $\widehat S$ are both parametrized by couples $(\Pi,\Theta)$, where  $\Pi:\CO^{2n}\to\CO^{2n}$ is an orthogonal projector and $\Theta:\ran(\Pi)\to\ran(\Pi)$ is symmetric. The extensions $\D_{\Pi,\Theta}$ and $\H_{\Pi,\Theta}$ have resolvents  
$$
(-\D_{\Pi,\Theta}+z)^{-1}=(-\D+z)^{-1}+G_{z}\Pi(\Theta-\Pi\,\tau G_{z}\Pi)^{-1}\Pi G_{\bar z}^{*}\,,\qquad z\in\varrho(\D_{\Pi,\Theta})\cap\varrho(\D)
$$  
$$
(-\H_{\Pi,\Theta}+z)^{-1}=(-\H+z)^{-1}+\widehat G_{z}\Pi(\Theta-\Pi\,\widehat \tau \widehat G_{z}\Pi)^{-1}\Pi \widehat G_{\bar z}^{*}\,,\qquad z\in\varrho(\H_{\Pi,\Theta})\cap\varrho(\H)\,.
$$
Moreover,
$$
\dom(\D_{\Pi,\Theta})=\{\Psi\in L^{2}(\RE ;\CO^{2}):\Psi=\Psi_{z}+G_{z}\xi\,,\ \Psi_{z}\in H^{1}(\RE ;\CO^{2})\,,\ \xi\in\ran(\Pi)\,,\ \Pi\tau\Psi=\Theta\xi\}
$$
$$
(-\D_{\Pi,\Theta}+z)\Psi=(-\D+z)\Psi_{z}\,,
$$
$$
\dom(\H_{\Pi,\Theta})=\{\psi\in L^{2}(\RE ):\psi=\psi_{z}+\widehat G_{z}\xi\,,\ \psi_{z}\in H^{2}(\RE )\,,\ \xi\in\ran(\Pi)\,,\ \Pi\widehat \tau\psi=\Theta\xi\},
$$
$$
(-\H_{\Pi,\Theta}+z)\psi=(-\H+z)\psi_{z}\,;
$$
such representations are $z$-independent and the decompositions of $\Psi$ in $\dom(\D_{\Pi,\Theta})$ and of $\psi$ in $
\dom(\H_{\Pi,\Theta})$ are unique.
\end{theorem} 
\begin{proof} The statements regarding the resolvents and the actions of the extensions follow from \cite[Theorem 2.1]{O&M} with $\Gamma_{\Pi,\Theta}(z)$ there defined either as $\Gamma_{\Pi,\Theta}(z):=\Theta-\Pi\tau G_{z}\Pi$ or as $\Gamma_{\Pi,\Theta}(z):=\widehat\Theta-\Pi\widehat\tau \widehat G_{z}\Pi$.\par
As regards the operators domains, we give the proof only for $\D_{\Pi,\Theta}$, since the one for $\H_{\Pi,\Theta}$ is of the same kind.
By the resolvent formula, one has 
$$\dom(\D_{\Pi,\Theta})=\{\Psi\in L^{2}(\RE;\CO^{2}):\Psi=\Psi_{z}+G_{z}\Pi(\Theta-\Pi\tau G_{z}\Pi)^{-1}\Pi\tau\Psi_{z}\,,\ \Psi_{z}\in H^{1}(\RE;\CO^{2})\}\,.
$$
Let us define $\xi_{z}:=(\Theta-\Pi\,\tau G_{z}\Pi)^{-1}\Pi\tau\Psi_{z}\in\ran(\Pi)$;  it is not difficult to check that 
$\xi_{z}$ does not depend on $z$ and so $\Psi=\Psi_{z}+G_{z}\xi$. Then 
$$\Pi\tau\Psi-\Theta\xi=\Pi\tau\Psi_{z}+\Pi\tau G_{z}\xi-\Theta\xi=\Pi\tau\Psi_{z}-(\Theta-\Pi
\tau G_{z}\Pi)\xi=0\,.$$ \end{proof}
\begin{remark} Notice that the choice $\Pi=\zero$ gives the self-adjoint extensions $\D$ and $\H$. Therefore, in the following we always suppose $\Pi\not=\zero$\par 
\end{remark}

Since we want to extend the relation \eqref{Rel} to the case with point interactions, we also need to consider the self-adjoint extensions of $\widehat S^{\circ}\oneb$. There are no essential changes with respect to the case of $\CO$-valued functions, the only relevant one being that the defect indices increase to $(4n,4n)$. The result is of the same kind as in Theorem \ref{K1}. 
 \begin{theorem}\label{K2}
The set of the self-adjoint extensions of $\widehat S\oneb$ is parametrized by couples $(\widehat\Pi,\widehat\Theta)$, where  $\widehat\Pi:\CO^{4n}\to\CO^{4n}$ is an orthogonal projector and $\widehat\Theta:\ran(\widehat\Pi)\to\ran(\widehat\Pi)$ is symmetric. The extension $\widehat\H_{\Pi,\Theta}$ has resolvent 
\begin{align*}
(-\widehat\H_{\widehat\Pi,\widehat\Theta}+z)^{-1}
=(-\H+z)^{-1}\oneb+(\widehat G_{z}\oneb)\widehat\Pi(\widehat\Theta-\widehat\Pi(\widehat\tau \widehat G_{z}\oneb)\widehat\Pi)^{-1}\widehat\Pi (\widehat G_{\bar z}^{*}\oneb),\quad z\in\varrho(\widehat\H_{\widehat\Pi,\widehat\Theta})\cap\varrho(\H).
\end{align*}
Moreover,
\begin{align*}
\dom(\widehat \H_{\widehat \Pi,\widehat \Theta})
=
\{\Psi\in L^{2}(\RE ;\CO^{2}):\Psi=\Psi_{z}+(\widehat G_{z}\oneb)\widehat\xi,\ \Psi_{z}\in H^{2}(\RE ;\CO^{2}),\ \widehat\xi\in\ran(\widehat \Pi),\ \widehat \Pi(\widehat\tau\oneb)\Psi=\widehat \Theta\widehat\xi\,\},
\end{align*}
$$
(-\widehat\H_{\widehat\Pi,\widehat\Theta}+z)\Psi=(-\H+z)\oneb\Psi_{z}\,;
$$
such representation is $z$-independent and the decomposition of $\Psi$ in $\dom(\widehat\H_{\widehat\Pi,\widehat\Theta})$ is unique.
\end{theorem} 
\begin{remark}\label{block} By Theorems \ref{K1} and \ref{K2}, if both $\widehat\Pi$ and $\widehat\Theta$ are block diagonal, i.e.,  $\widehat\Pi=\Pi_{1}\oplus\Pi_{2}$ and $\widehat\Theta=\Theta_{1}\oplus\Theta_{1}$, then 
$$
(-\widehat \H_{\Pi_{1}\oplus\Pi_{2},\Theta_{1}\oplus\Pi_{2}}+z)^{-1}=(-\H_{\Pi_{1},\Theta_{1}}+z)^{-1}\oplus(-\H_{\Pi_{2},\Theta_{2}}+z)^{-1},
$$
equivalently, 
$$
\widehat \H_{\Pi_{1}\oplus\Pi_{2},\Theta_{1}\oplus\Pi_{2}}=
\H_{\Pi_{1},\Theta_{1}}\oplus\H_{\Pi_{2},\Theta_{2}}.
$$
In particular, 
 $$
 \widehat \H_{\Pi\uno,\Theta\uno}=\H_{\Pi,\Theta}\oneb\,.
 $$
 \end{remark}
\begin{remark}\label{rem-delta} Since $g_{z}$ is the fundamental solution of $-\D+z$, one has
$$
(-\D_{\Pi,\Theta}+z)\Psi=(-\D+z)(\Psi-G_{z}\xi)=(-\D+z)\Psi-\sum_{k=1}^{n}\xi_{k}\delta_{y_{k}}\,,
$$
i.e.,
$$
\D_{\Pi,\Theta}\Psi=\D\Psi+\sum_{k=1}^{n}\xi_{k}\delta_{y_{k}}\,,\quad\xi\equiv(\xi_{1},\dots,\xi_{n})\,,
$$
where the action of $\D$ on $\Psi\in L^{2}(\RE ;\CO^{2})$ is to be understood in distributional sense.
Analogously,
$$
\H_{\Pi,\Theta}\psi=\H\psi+\sum_{k=1}^{n}(\xi_{k,1}\delta_{y_{k}}+\xi_{k,2}\delta'_{y_{k}}),
\quad\xi\equiv((\xi_{1,1},\xi_{1,2}),\dots,(\xi_{n,1},\xi_{n,2}))\,,
$$
$$
\widehat \H_{\widehat \Pi,\widehat \Theta}\Psi=\H\oneb\Psi+\sum_{k=1}^{n}(\widehat\xi_{k,1}\delta_{y_{k}}+\widehat\xi_{k,2}\delta'_{y_{k}}),
\quad\widehat\xi\equiv((\widehat\xi_{1,1},\widehat\xi_{1,2}),\dots,(\widehat\xi_{n,1},\widehat\xi_{n,2}))\,.
$$
\end{remark} 
In the following, we use the abbreviated notations $\D_{\Theta}\equiv\D_{\uno,\Theta}$, $\H_{\Theta}\equiv\H_{\uno,\Theta}$, $\widehat\H_{\widehat\Theta}\equiv\widehat\H_{\uno,\widehat\Theta}\,.$  
\section{$\D^{2}=\H+\frac14$ with point interactions}\label{S3}
We begin this section by providing an equivalent representation of the domains and actions of the self-adjoint operators we built in Section \ref{S2}. In the next theorem and in the following, 
$$\D_{\RE\backslash Y}:{\mathscr D}'(\RE\backslash Y;\CO^{2})\to {\mathscr D}'(\RE\backslash Y;\CO^{2})\,,\qquad \H_{\RE\backslash Y}:{\mathscr D}'(\RE\backslash Y)\to {\mathscr D}'(\RE\backslash Y)
$$ denote the free Dirac and Schr\"odinger operators in the space of distributions on $\RE\backslash Y$; their restrictions to $H^{1}(\RE\backslash Y;\CO^{2})$ and $H^{2}(\RE\backslash Y)$ are $L^{2}(\RE;\CO^{2})$ and $L^{2}(\RE)$-valued respectively.
\begin{theorem}\label{K3} Let $\D_{\Pi,\Theta}$, $\H_{\Pi,\Theta}$ and $\widehat \H_{\widehat \Pi,\widehat \Theta}$ as in Section \ref{S2}. Then
$$
\dom(\D_{\Pi,\Theta})=\{\Psi\in H^{1}(\RE\backslash Y;\CO^{2}):\rho\Psi\in\ran(\Pi),\ \Pi\tau\Psi=\Theta\rho\Psi\},\quad \D_{\Pi,\Theta}\Psi=\D_{\RE\backslash Y}\Psi\,,
$$
$$
\dom(\H_{\Pi,\Theta})=\{\psi\in H^{2}(\RE\backslash Y):\widehat\rho\psi\in\ran(\Pi),\ \Pi\widehat\tau\psi=\Theta\widehat\rho\psi\},\quad \H_{\Pi,\Theta}\psi=\H_{\RE\backslash Y}\psi\,,
$$
$$
\dom(\widehat\H_{\widehat\Pi,\widehat\Theta})=\{\Psi\in H^{2}(\RE\backslash Y;\CO^{2}):(\widehat\rho\uno)\Psi\in\ran(\widehat\Pi),\ \widehat\Pi(\widehat\tau\uno)\Psi=\widehat\Theta(\widehat\rho\uno)\Psi\},\quad \widehat\H_{\widehat\Pi,\widehat\Theta}\Psi=(\H_{\RE\backslash Y}\uno)\Psi\,,
$$
where 
\begin{align*}
			\rho:H^1(\RE\backslash Y;\CO^{2})\to\CO^{2n}\,,\quad \rho\Psi:=\big(
			\rho_{y_{1}}\Psi\,,\dots,\rho_{y_{n}}\Psi\big)\,,\qquad\rho_{y}\Psi:=i\sigma_{1}[\Psi]_{y}\,,
			\label{rho}
		\end{align*}
	$$
		\widehat\rho:H^2(\RE\backslash Y)\to\CO^{2n}\,,\qquad \widehat\rho \psi:=
		\big(\widehat\rho_{y_{1}}\psi\,,\dots,\widehat\rho_{y_{n}}\psi\big)\,,\qquad
		\widehat\rho_{y}\psi:=
		[\psi'\,]_{y}\oplus[-\psi]_{y}\,,
	$$
		
	$$	[f]_{y}:=f(y^{+})-f(y^{-})\,.
	$$

\end{theorem} 
\begin{proof}  Let $\Psi=\Psi_{z}+G_{z}\xi\in \dom(\D_{\Pi,\Theta})$. One has 
$\Psi_{z}\in H^{1}(\RE;\CO^{2})\subset H^{1}(\RE\backslash Y;\CO^{2})$ and $G_{z}\xi\in H^{1}(\RE\backslash Y;\CO^{2})$; therefore, $\Psi\in H^{1}(\RE\backslash Y;\CO^{2})$. By $[G_{z}\xi]_{y}=i\sigma_{1}\xi$, one gets $\rho G_{z}\xi=\xi$; 
furthermore, by $H^{1}(\RE;\CO^{2})\subset C_{b}(\RE;\CO^{2})$, one gets $\rho\Psi_{z}=0$. Therefore, 
$$\dom(\D_{\Pi,\Theta})\subseteq {\mathcal D}:=\{\Psi\in H^{1}(\RE\backslash Y;\CO^{2}):\rho\Psi\in\ran(\Pi),\ \Pi\tau\Psi=\Theta\rho\Psi\}\,.$$
By Remark \ref{rem-delta}, $\D_{\Pi,\Theta}\Psi=\D_{\RE\backslash Y}\Psi$ for any $\Psi\in \dom(\D_{\Pi,\Theta})$, i.e., $\D_{\Pi,\Theta}\subset \D_{\RE\backslash Y}|{\mathcal D}$. Moreover, by integration by parts, $\D_{\RE\backslash Y}|{\mathcal D}$ is symmetric; hence, since $\D_{\Pi,\Theta}$ is self-adjoint, one gets $\D_{\Pi,\Theta}=\D_{\RE\backslash Y}|{\mathcal D}$.\par 
The proofs for $\H_{\Pi,\Theta}$ and 
$\widehat\H_{\widehat\Pi,\widehat\Theta}$ are of the same kind, using the relation $\widehat\rho \widehat G_{z}\xi=\xi$.
\end{proof}
\begin{remark}\label{H1H2} Notice that $\psi\in H^{1}(\RE\backslash Y)$ belongs to $H^{1}(\RE)$ if and only if 
$[\psi]_{y_{k}}=0$ for any $k$ and consequently $\psi\in H^{2}(\RE\backslash Y)$ belongs to $H^{2}(\RE)$ if and only if $[\psi]_{y_{k}}=[\psi']_{y_{k}}=0$ for any $k$.
\end{remark}
By Theorem \ref{K3} and by
$$
(\D_{\RE\backslash Y})^{2}=\left(\H_{\RE\backslash Y}+\frac14\right)\uno\,,
$$
given the couple $(\Pi,\Theta)$, one gets that  the couple $(\widehat\Pi,\widehat\Theta)$ is such that 
\be\label{D2=H}
(\D_{\Pi,\Theta})^{2}=\widehat\H_{\widehat\Pi,\widehat\Theta}+\frac{\uno}4\,,
\ee
if and only if  
\be\label{dom=dom}
\dom((\D_{\Pi,\Theta})^{2})=\dom(\widehat\H_{\widehat\Pi,\widehat\Theta})\,.
\ee
Therefore, exploiting the definitions of the operator domains in Theorem \ref{K3}, there exists a couple $(\widehat\Pi,\widehat\Theta)$ for which \eqref{D2=H} holds if and only if,  given $(\Pi,\Theta)$, there exists $(\widehat\Pi,\widehat\Theta)$, $\widehat\Pi$ an orthogonal projector in $\CO^{4n}$ and $\widehat\Theta$ symmetric in $\ran(\widehat\Pi)$, such that 
\be\label{iff}
\begin{cases}\rho\Psi\oplus\rho\D_{\RE\backslash Y}\Psi\in\ran(\Pi\oplus\Pi)\\
(\Pi\oplus\Pi)\tau\Psi\oplus\tau\D_{\RE\backslash Y}\Psi=(\Theta\oplus\Theta)\rho\Psi\oplus\rho\D_{\RE\backslash Y}\Psi
\end{cases}\quad\iff\qquad
\begin{cases}(\widehat\rho\uno)\Psi\in\ran(\widehat\Pi)\\
\widehat\Pi(\widehat\tau\uno)\Psi=\widehat\Theta(\widehat\rho\uno)\Psi\,.
\end{cases}
\ee
\subsection{Spectral correspondence} The relation \eqref{D2=H} entails $\pm z\in\varrho(\D_{\Pi,\Theta})$ if and only if $z^{2}-\frac14\in\varrho(\widehat\H_{\widehat\Pi,\widehat\Theta})$, equivalently,
$\pm \lambda\in\sigma(\D_{\Pi,\Theta})$ if and only if $\lambda^{2}-\frac14\in\sigma(\widehat\H_{\widehat\Pi,\widehat\Theta})$, and 
\be\label{res-ext}
(-\D_{\Pi,\Theta}+z)^{-1}=(\D_{\Pi,\Theta}+z)\left(-\widehat\H_{\widehat\Pi,\widehat\Theta}+\left(z^{2}-\frac14\right)\uno\right)^{\!\!-1}\,.
\ee
Furthermore, since, by the invariance of the essential spectrum by finite-rank perturbations, 
$$
\sigma_{ess}(\widehat\H_{\widehat\Pi,\widehat\Theta})=\sigma_{ess}(\H\uno)=[0,\infty)\,,\qquad\sigma_{ess}(\D_{\widehat\Pi,\widehat\Theta})=\sigma_{ess}(\D)=\left(-\infty,-\frac12\right]\cup \left[\frac12,+\infty\right)\,,
$$ 
one gets 
$$
\lambda\in\sigma_{disc}(\widehat\H_{\widehat\Pi,\widehat\Theta})\cap\left[-\frac14,0\right)
\quad\iff\quad\pm\left(\lambda+\frac14\right)^{\frac12}\in\sigma_{disc}(\D_{\Pi,\Theta})\,.
$$
By the resolvent formulae in Theorems \ref{K1} and \ref{K2}, 
\be\label{detD}
\lambda\in \sigma_{disc}(\D_{\Pi,\Theta})\quad\iff\quad\text{$\lambda\in(-1/2,1/2)\quad$ and $\quad\det(\Theta-\Pi\tau G_{\lambda}\Pi)=0$}\,,
\ee
\be\label{detH}
\lambda\in \sigma_{disc}(\widehat\H_{\widehat\Pi,\widehat\Theta})\quad\iff\quad\text{$\lambda\in(-\infty,0)\quad$ and $\quad\det(\widehat\Theta-\widehat\Pi(\widehat\tau \widehat G_{\lambda}\uno)\widehat\Pi)=0$}\,.
\ee
Now, we solve \eqref{iff} starting from the simplest case $n=1$, $\Pi=\uno$ and then proceeding step-by-step towards the most general case.
\subsection{The case $n=1$, $\Pi=\uno$.}\label{n=1-pi=1} By \eqref{iff}, given the $2\times 2$ Hermitian  matrix $\Theta$, we need to find the $4\times 4$ Hermitian  matrix $\widehat\Theta$ such that 
\begin{equation}\label{iff-n=1}
	\begin{bmatrix} \tau_{y} \Psi \\ \tau_{y} \D_{\RE\backslash \{y\}}\Psi	\end{bmatrix}=
	\begin{bmatrix} \Theta&{\mathbb 0} \\ {\mathbb 0}&\Theta\end{bmatrix}\begin{bmatrix} \rho_{y} \Psi \\ \rho_{y} \D_{\RE\backslash \{y\}}\Psi\end{bmatrix}
	\quad \iff \quad(\widehat \tau_{y} \uno)\Psi = \widehat\Theta(\widehat \rho_{y}\uno) \Psi\,.
\end{equation}
To solve \eqref{iff-n=1}, at first we look for the two invertible matrices $M_1$ and $M_2$ such that
\begin{align}
(\widehat \tau_{y} \uno)\Psi = M_1 \begin{bmatrix} \tau_{y} \Psi \\ \tau_{y} \D_{\RE\backslash \{y\}}\Psi	\end{bmatrix}, \qquad (\widehat \rho_{y}\uno) \Psi = M_2 \begin{bmatrix} \rho_{y} \Psi \\ \rho_{y} \D_{\RE\backslash \{y\}}\Psi	\end{bmatrix}.
\label{ref1}
\end{align}
By direct calculations, one gets 
\be\label{mat-1-2}
M_1 = \begin{bmatrix}	1 & 0 & 0 & 0 \\ 0 & \frac{i}{2} & 0 & i \\ 0 & 1 & 0 & 0 \\- \frac{i}{2}& 0 & i & 0\end{bmatrix}, \qquad 
M_2 = \begin{bmatrix} \frac{1}{2} & 0 & 1 & 0 \\ 0 & i & 0 & 0 \\ 0 & -\frac{1}{2} & 0 & 1 \\ i & 0 & 0 & 0\end{bmatrix}.
\ee
Therefore, \eqref{iff-n=1} rewrites as
$$		M_1^{-1}(\widehat \tau_{y}\uno) \Psi = (\Theta \oplus \Theta) M_2^{-1} (\widehat \rho_{y}\uno) \Psi
\quad \iff \quad (\widehat \tau_{y}\uno) \Psi = \widehat\Theta(\widehat \rho_{y}\uno) \Psi
$$
and so the relation between $\widehat\Theta $ and $\Theta$ is given by
\begin{align}\label{wT1}
	\widehat\Theta =& M_1 (\Theta\oplus \Theta) M_2^{-1}\,.
\end{align}
By
$$
\widehat\Theta=\widehat\Theta^{*}\quad\iff\quad M^{*}_1 M_2 (\Theta\oplus \Theta)=
(\Theta\oplus \Theta)M_{2}^{*}M_{1},
$$
$\widehat\Theta$ is symmetric by the relations 
\be\label{M1M2}
M_{1}^{*}M_{2}=\begin{bmatrix}\zero&\uno\\ \uno&\zero
\end{bmatrix}=M_{2}^{*}M_{1}\,.
\ee
More explicitly, if 
$$
\Theta=\begin{bmatrix} a& b\\ \bar b & d \end{bmatrix},\qquad a,d\in\RE,\, b\in\CO\,,
$$ 
then $\widehat\Theta$ is represented by the Hermitian  matrix
$$
	\widehat\Theta = \begin{bmatrix} 0 & -ib & 0 & -ia \\ i\bar b & d & id & 0 \\ 0 & -id & 0 & -i\bar b \\ia& 0 & ib & -a\end{bmatrix}\,.
$$
If $a=d=0$ and $b\in\RE$, i.e., if $\Theta=b\sigma_{1}$, then $\widehat\Theta = b(\sigma_{2}\oplus \sigma_{2})\equiv b\sigma_{2}\uno$, where $\sigma_{2}$ denotes the Pauli matrix
$$
\sigma_2=
		\begin{bmatrix}
			0 & -i\\
			i & 0\\
		\end{bmatrix},
$$
and, by Remark \ref{block}, the corresponding Schr\"odinger operator in $L^{2}(\RE;\CO^{2})$ is  block diagonal: 
\be\label{bd}
(\D_{b\sigma_{1}})^{2} = \left(\H_{b\sigma_{2}} + \frac{1}{4}\right)\!\oneb\,.
\ee
\subsection{The case $n=1$, $\Pi\not=\oneb$.} Here we take $\Pi : \C^2 \to \C^2$ a not trivial orthogonal projection, i.e., $\text{dim}(\ran(\Pi))=1$, and $\Theta:\ran(\Pi)\to\ran(\Pi)$ identifies with the multiplication by $\theta\in\RE$. By \eqref{ref1}, \eqref{iff} rewrites as 
\be\label{iff-pi}
\begin{cases}
(\widehat \rho_{y}\uno) \Psi \in \ran(M_2(\Pi \oplus \Pi))\\
M_2(\Pi \oplus \Pi)M_1^{-1}(\widehat \tau_{y}\uno)\Psi= \theta (\widehat \rho_{y}\uno)\Psi
\end{cases}
\quad\iff\qquad
\begin{cases}(\widehat\rho_{y}\uno)\Psi\in\ran(\widehat\Pi)\\
\widehat\Pi(\widehat\tau_{y}\uno)\Psi=\widehat\Theta(\widehat\rho_{y}\uno)\Psi\,.
\end{cases}
\ee
Therefore, $\widehat\Pi:\CO^{4}\to\CO^{4}$ is the orthogonal projection onto the $2$-dimensional subspace 
$$\ran(\widehat\Pi)=  \ran(M_2(\Pi \oplus \Pi))=\ran(M_2(\Pi \oplus \Pi)M_1^{-1})\,,
$$ 
i.e.,
\begin{align*}
\widehat\Pi=
&M_{2}(\Pi\oplus\Pi)((\Pi\oplus\Pi)M_{2}^{*}M_{2}(\Pi\oplus\Pi))^{-1}(\Pi\oplus\Pi)M_{2}^{*}\\=&M_{2}(\Pi\oplus\Pi)(M_{2}^{*}M_{2})^{-1}(\Pi\oplus\Pi)M_{2}^{*}\\=&
(M_{2}(\Pi\oplus\Pi)M_{2}^{-1})(M_{2}(\Pi\oplus\Pi)M_{2}^{-1})^{*}
\,.
\end{align*}
By \eqref{M1M2}, $
	M_2(\Pi \oplus \Pi)M_1^{-1}$ is symmetric. Hence, $\ran(\widehat\Pi)=\ker(M_2(\Pi \oplus \Pi)M_1^{-1})^{\perp}$ and the symmetric operator 
$$
M_2(\Pi \oplus \Pi)M_1^{-1}:\ran(\widehat\Pi)\to\ran(\widehat\Pi)
$$
is a bijection. Then, \eqref{iff-pi} gives
$$	\widehat \Theta: \ran(\widehat \Pi) \to \ran(\widehat \Pi), \qquad \widehat \Theta := \theta(M_2(\Pi \oplus \Pi)M_1^{-1})^{-1}.
	$$	
\subsection{The case $n>1$, $\Pi=\uno$.} In order to exploit the results from the $n=1$ case, 
we introduce the unitary operator 
\be\label{uni}
U:\CO^{4n}\to\CO^{4n}\,,\qquad U(\xi_{1},\xi_{2},\dots,\xi_{2n}):=(\xi_{1},\xi_{n+1},\xi_{2}, \xi_{n+2},\dots,\xi_{n},\xi_{2n})\,,\quad \xi_{k}\in\CO^{2}\,.
\ee
By such a definition, 
$$
U(\tau\Psi\oplus\tau\D_{\RE\backslash Y}\Psi)=\left(\,\begin{bmatrix} \tau_{y_{1}} \Psi \\ \tau_{y_{1}} \D_{\RE\backslash Y}\Psi	\end{bmatrix},\dots,\begin{bmatrix} \tau_{y_{n}} \Psi \\ \tau_{y_{n}} \D_{\RE\backslash Y}\Psi	\end{bmatrix}\,\right),
$$
$$
U(\rho\Psi\oplus\rho\D_{\RE\backslash Y}\Psi)=\left(\,\begin{bmatrix} \rho_{y_{1}} \Psi \\ \rho_{y_{1}} \D_{\RE\backslash Y}\Psi	\end{bmatrix},\dots,\begin{bmatrix} \rho_{y_{n}} \Psi \\ \rho_{y_{n}} \D_{\RE\backslash Y}\Psi	\end{bmatrix}\,\right)\,.
$$
Therefore, setting
$$
M_{1}^{\oplus}:\CO^{4n}\to\CO^{4n}\,,\qquad M_{1}^{\oplus}:=M_{1}\oplus\dots\oplus M_{1}\,,
$$
$$
M_{2}^{\oplus}:\CO^{4n}\to\CO^{4n}\,,\qquad M_{2}^{\oplus}:=M_{2}\oplus\dots\oplus M_{2}\,,
$$
by \eqref{ref1}, one gets
$$
M_{1}^{\oplus}U(\tau\Psi\oplus\tau\D_{\RE\backslash Y}\Psi)
=\big((\widehat\tau_{y_{1}}\uno)\Psi,\dots,(\widehat\tau_{y_{n}}\uno)\Psi\big)=U(\widehat\tau\uno)\Psi\,,
$$ 
$$
M_{2}^{\oplus}U(\rho\Psi\oplus\rho\D_{\RE\backslash Y}\Psi)
=\big((\widehat\rho_{y_{1}}\uno)\Psi,\dots,(\widehat\rho_{y_{n}}\uno)\Psi\big)=U
(\widehat\rho\uno)\Psi
$$ 
and so \eqref{iff} rewrites as
$$
U^{*}(M_{1}^{\oplus})^{-1}U(\widehat\tau\uno)\Psi=(\Theta\oplus\Theta)
U^{*}(M_{2}^{\oplus})^{-1}U(\widehat\rho\uno)\Psi\quad \iff \quad(\widehat \tau\uno)\Psi = \widehat\Theta(\widehat \rho\uno) \Psi\,.
$$
This gives
\be\label{wT}
\widehat\Theta=U^{*}M_{1}^{\oplus}U(\Theta\oplus\Theta)
U^{*}(M_{2}^{\oplus})^{-1}U\,.
\ee
Such a operator $\widehat\Theta$ is symmetric by 
\be\label{M1M2-n}
U^{*}(M_{1}^{\oplus})^{*}M_{2}^{\oplus}U=\begin{bmatrix}\zero&\uno\\ \uno&\zero
\end{bmatrix}=U^{*}(M_{2}^{\oplus})^{*}M^{\oplus}_{1}U\,.
\ee
The relations \eqref{M1M2-n} generalize \eqref{M1M2}, since $U=\uno$ whenever $n=1$, and are a consequence of $\eqref{M1M2}$ itself and the definition \eqref{uni}.
\subsection{The case $n>1$, $\Pi\not=\oneb$.} Finally, we consider the most general case. Using the unitary $U:\CO^{4n}\to\CO^{4n}$ as in the previous section, \eqref{iff} rewrites as
\be\label{iff-pi-n}
\begin{cases}
U^{*}(M^{\oplus}_{2})^{-1}U(\widehat\rho\uno)\Psi\in\ran(\Pi\oplus\Pi)\\
(\Pi\oplus\Pi)U^{*}(M_{1}^{\oplus})^{-1}U(\widehat\tau\uno)\Psi=(\Theta\oplus\Theta)
U^{*}(M_{2}^{\oplus})^{-1}U(\widehat\rho\uno)\Psi\end{cases}
\iff \quad
\begin{cases}(\widehat\rho\uno)\Psi\in\ran(\widehat\Pi)\\
\widehat\Pi(\widehat\tau\uno)\Psi=\widehat\Theta(\widehat\rho\uno)\Psi\,.
\end{cases}
\ee
This gives the orthogonal projector $\widehat\Pi:\CO^{4n}\to\CO^{4n}$, with  $\text{dim}(\ran(\widehat\Pi))=2\,\text{dim}(\ran(\Pi))$, such that
\be\label{wPp}
\ran(\widehat\Pi)=\ran(U^{*}M_{2}^{\oplus}U(\Pi\oplus\Pi))=\ran(U^{*}M_{2}^{\oplus}U(\Pi\oplus\Pi)U^{*}(M_{1}^{\oplus})^{-1}U)\,,
\ee 
i.e.,
\begin{align*}
\widehat\Pi=&(U^{*}M^{\oplus}_{2}U(\Pi\oplus\Pi))
\big((U^{*}M_{2}^{\oplus}U(\Pi\oplus\Pi))^{*}(U^{*}M_{2}^{\oplus}U(\Pi\oplus\Pi))\big)^{-1}
(U^{*}M^{\oplus}_{2}U(\Pi\oplus\Pi))^{*}\nonumber\\
=&U^{*}M^{\oplus}_{2}U(\Pi\oplus\Pi)((U^{*}M^{\oplus}_{2}U)^{*}U^{*}M^{\oplus}_{2}U)^{-1}(\Pi\oplus\Pi)U^{*}(M^{\oplus}_{2})^{*}U\label{wPpp}\\
=&\big(U^{*}M^{\oplus}_{2}U(\Pi\oplus\Pi)U^{*}(M^{\oplus}_{2})^{-1}\big)\big(U^{*}M^{\oplus}_{2}U(\Pi\oplus\Pi)U^{*}(M^{\oplus}_{2})^{-1}\big)^{*}\nonumber
\,,
\end{align*}
and $\widehat\Theta:\ran(\widehat\Pi)\to\ran(\widehat\Pi)$,
\be\label{wTp}
\widehat\Theta:=\big(U^{*}M_{2}^{\oplus}U
(\Pi\oplus\Pi)U^{*}(M_{1}^{\oplus})^{-1}U\big)^{-1}
U^{*}M_{2}^{\oplus}U(\Theta\oplus\Theta)U^{*}(M_{2}^{\oplus})^{-1}U.
\ee
By \eqref{M1M2-n},  $U^{*}M^{\oplus}_2U(\Pi \oplus \Pi)U^{*}(M_1^{\oplus})^{-1}U$
is symmetric. Therefore, by $$\ran(\widehat\Pi)=\ran(U^{*}M^{\oplus}_2U(\Pi \oplus \Pi)U^{*}(M_1^{\oplus})^{-1}U)=\ker(U^{*}M^{\oplus}_2U(\Pi \oplus \Pi)U^{*}(M_1^{\oplus})^{-1}U)^{\perp}\,,$$ the operator 
$$
U^{*}M^{\oplus}_2U(\Pi \oplus \Pi)U^{*}(M_1^{\oplus})^{-1}U:\ran(\widehat\Pi)\to\ran(\widehat\Pi)
$$
is a bijection and $\widehat\Theta$ is well defined. To conclude, we have to show that $\widehat\Theta $ is symmetric. By \eqref{wTp} and by $U^{*}U=\uno$,
$$
\text{$\widehat\Theta$ is symmetric} \iff\text{$M_{2}^{\oplus}U(\Theta\oplus\Theta)
(\Pi\oplus\Pi)U^{*}(M_{1}^{\oplus})^{-1}$ is symmetric}
$$
and so $\widehat\Theta$ is symmetric by \eqref{M1M2-n}.
\begin{remark}\label{equi} Let us point out that it is not necessary to determine  $\widehat \Pi$ and $\widehat \Theta$ explicitly in order to write down the domain of $\widehat\H_{\widehat\Pi,\widehat\Theta}$. Indeed, by \eqref{iff},  
\begin{align*}
\dom(\widehat\H_{\widehat\Pi,\widehat\Theta})=\big\{\Psi\in H^{2}(\RE\backslash Y;\CO^{2}):
\rho\Psi\oplus\rho\D_{\RE\backslash Y}\Psi\in\ran(\Pi\oplus\Pi)&\\
(\Pi\tau\Psi)\oplus(\Pi\tau\D_{\RE\backslash Y}\Psi)=(\Theta\rho\Psi)\oplus(\Theta\rho\D_{\RE\backslash Y}\Psi)&\big\}\,.
\end{align*}
However, one needs to know $\widehat \Pi$ and $\widehat \Theta$ in order to write down the resolvent of $\widehat\H_{\widehat\Pi,\widehat\Theta}$, according to Theorem \ref{K2}.\par
The above representation of $\dom(\widehat\H_{\widehat\Pi,\widehat\Theta})$ suggests an alternative way to build the self-adjoint extensions of $\widehat S^{\circ}\uno=\H\uno|C^{\infty}_{comp}(\RE;\CO^{2})$: one can apply the results in \cite{JFA} and \cite{O&M} to $\H\uno|\ker(\widetilde\tau)$, where 
$$
\widetilde\tau:H^{2}(\RE;\CO^{2})\to\CO^{2n}\oplus\CO^{2n}\,,\qquad \widetilde\tau\Psi:=\tau\Psi\oplus\tau\D\Psi\,.$$
In that case, the family of self-adjoint extensions of $\widehat S^{\circ}\uno$ is represented by  operators of the kind $\widetilde H_{\widetilde \Pi,\widetilde\Theta}$, where $\widetilde\Pi$ is an orthogonal projector in $\CO^{2n}\oplus\CO^{2n}$ and $\widetilde\Theta$ is a symmetric operator in $\ran(\widetilde\Theta)$. With respect to this parametrization, one has that $\D_{\Pi,\Theta}^{2}=
\widetilde \H_{\widetilde \Pi,\widetilde\Theta}+\frac{\uno}4$ if and only if 
$\widetilde\Pi=\Pi\oplus\Pi$ and $\widetilde\Theta=\Theta\oplus\Theta$. Even if such a  correspondence is more explicit than the one which uses the couple $(\widehat\Pi,\widehat\Theta)$, it has the drawback that it works with a representation of the family of self-adjoint extensions of $\widehat S^{\circ}\uno$ which is different from the usual one and which lacks of  the analogous of the property $\widehat \H_{\Pi\uno,\Theta\uno}=\H_{\Pi,\Theta}\uno$. Therefore, in this paper we prefer to work with the family $\widehat \H_{\widehat \Pi,\widehat \Theta}$.
\end{remark}
\begin{remark}\label{equi2} Suppose that for any $\Psi\equiv(\psi_{1},\psi_{2})\in\dom(\D_{\Pi,\Theta})$ one has  
$$
\begin{cases}
\rho\Psi\in\ran(\Pi)\\
\Pi\tau\Psi=\Theta\rho\Psi
\end{cases}\iff\quad
\begin{cases}
B_{1}(\psi_{1})=0\\
B_{2}(\psi_{2})=0\,,
\end{cases}
$$
with some linear operators $B_{1}:H^{1}(\RE\backslash Y)\to\CO^{d_{1}}$ and $B_{2}:H^{1}(\RE\backslash Y)\to\CO^{d_{2}}$. Then, by the representation of $\dom(\widehat\H_{\widehat\Pi,\widehat\Theta})$ in Remark \ref{equi}, there follows that the boundary conditions for $\widehat\H_{\widehat\Pi,\widehat\Theta}$ rewrites as
$$
\begin{cases}
B_{1}(\psi_{1})=0\\
B_{1}\big(-i\psi'_{2}+\frac12\,\psi_{1}\big)=0\\
B_{2}(\psi_{2})=0\\
B_{2}\big(-i\psi'_{1}-\frac12\,\psi_{2}\big)=0
\end{cases}
\equiv\quad
\begin{cases}
B_{1}(\psi_{1})=0\\
B_{2}(\psi_{1}')=0\\
B_{2}(\psi_{2})=0\\
B_{1}(\psi'_{2})=0\,.
\end{cases}
$$
This entails 
$$
\dom(\widehat\H_{\widehat\Pi,\widehat\Theta})=\dom(\H_{1,2})\oplus\dom(\H_{2,1})\,,\qquad 
(\D_{\Pi,\Theta})^{2}=\left(\H_{1,2}+\frac14\right)\oplus\left(
\H_{2,1}+\frac14\right)\,,
$$ 
where the self-adjoint operators $\H_{j,k}:\dom(\H_{j,k})\subseteq L^{2}(\RE)\to L^{2}(\RE)$ are defined by 
$$
\dom(\H_{1,2}):=\{\psi\in H^{2}(\RE\backslash Y): B_{1}(\psi)=0,\ B_{2}(\psi')=0\}\,,\quad \H_{1,2}\psi:=\H_{\RE\backslash Y}\psi\,.
$$
$$
\dom(\H_{2,1}):=\{\psi\in H^{2}(\RE\backslash Y): B_{2}(\psi)=0,\ B_{1}(\psi')=0\}\,,\quad \H_{2,1}\psi:=\H_{\RE\backslash Y}\psi\,.
$$
\end{remark}
\section{Applications}\label{S4}
\subsection{Local boundary conditions.}\label{LBC} Here we consider the case corresponding to local boundary conditions for the Dirac operator, i.e., boundary conditions which do not couple the values of $\Psi$ at different point. That means 
$$
\Pi=\Pi_{1}\oplus\dots\oplus\Pi_{n}\,,  
\qquad\Pi_{k}:\CO^{2}\to\CO^{2}\,, \qquad 1\le k\le n\,,
$$
$$\Theta=\Theta_{1}\oplus\dots\oplus\Theta_{n}\,,  
\qquad\Theta_{k}:\ran(\Pi_{k})\to\ran(\Pi_{k})\,, \qquad 1\le k\le n\,.
$$ In this case, by
$$
U((\Pi_{1}\oplus\dots\oplus\Pi_{n})\oplus (\Pi_{1}\oplus\dots\oplus\Pi_{n}))U^{*}
=(\Pi_{1}\oplus\Pi_{1})\oplus\dots\oplus(\Pi_{n}\oplus \Pi_{n})\,,
$$ 
one gets, by \eqref{wPp},
\begin{align*}
\ran(\widehat\Pi)=&\ran(U^{*}(M_{2}(\Pi_{1}\oplus\Pi_{1})M_{1}^{-1})\oplus\dots\oplus
(M_{2}(\Pi_{n}\oplus\Pi_{n})M_{1}^{-1})U)\\
=&\ran(U^{*}(\widehat\Pi_{1}\oplus\dots\oplus\widehat\Pi_{n})U)\,,
\end{align*}
where, in the case  $\Pi_{k}\not=\uno$, $\widehat\Pi_{k}:\CO^{4}\to\CO^{4}$ 
is the orthogonal projector onto the $2$-dimensional subspace 
\be\label{lbc1}
\ran(\widehat\Pi_{k})=\ran(M_{2}(\Pi_{k}\oplus\Pi_{k})M_{1}^{-1})\,,
\ee
otherwise $\widehat\Pi_{k}=\uno$. 
Then, by \eqref{wTp} and by
\begin{align*}
&U((\Pi_{1}\Theta_{1}\Pi_{1}\oplus\dots\oplus\Pi_{n}\Theta_{n}\Pi_{n})\oplus (\Pi_{1}\Theta_{1}\Pi_{1}\oplus\dots\oplus\Pi_{n}\Theta_{n}\Pi_{n}))U^{*}\\
=&(\Pi_{1}\Theta_{1}\Pi_{1}\oplus\Pi_{1}\Theta_{1}\Pi_{1})\oplus\dots\oplus
(\Pi_{n}\Theta_{n}\Pi_{n}\oplus\Pi_{n}\Theta_{n}\Pi_{n})\,,
\end{align*}
one obtains
\begin{align*}
\widehat\Theta=
U^{*}(\widehat\Theta_{1}\oplus\dots\oplus\widehat\Theta_{n})U\,,
\end{align*}
where, in the case  $\Pi_{k}\not=\uno$, $\Theta_{k}\in\RE$,
$$
\widehat\Theta_{k}:\ran(\widehat\Pi_{k})\to\ran(\widehat\Pi_{k})\,,\qquad \widehat\Theta_{k}=
\Theta_{k}M_{1}(\Pi_{k}\oplus\Pi_{k})M_{2}^{-1}\,,
$$
otherwise,
$$
\widehat\Theta_{k}:\CO^{4}\to\CO^{4}\,,\qquad \widehat\Theta_{k}=
M_{1}(\Theta_{k}\oplus\Theta_{k})M_{2}^{-1}\,.
$$
Therefore, the corresponding boundary conditions for $\widehat\H_{\widehat\Pi,\widehat\Theta}$ are 
$$
\widehat\rho_{y_{k}}\Psi\in\ran(\widehat\Pi_{k})\,,\qquad \widehat\Pi_{k}\widehat\tau_{y_{k}}\Psi=\widehat\Theta_{k}\widehat\rho_{y_{k}}\Psi\,,\qquad 1\le k\le n\,,
$$
and so they are local as well.\par
\subsection{Gesztesy-\v{S}eba realizations}\label{GS} These two families of self-adjoint realizations of the Dirac operator with local point interactions correspond, in the non relativistic limit, to Schr\"odinger operators with local point interactions either of $\delta$-type or of $\delta'$-type (see \cite{GS}, \cite[Appendix J]{AGHKH}, \cite{CMP}). The operators in the $\alpha$-family have self-adjointness domains 
\be\label{alpha}
\dom(\D_{\alpha})=\{\Psi\equiv(\psi_{1},\psi_{2})\in H^{1}(\RE)\oplus 
H^{1}(\RE\backslash Y):[\psi_{2}]_{y_{k}}=-i\alpha_{k}\psi_{1}(y_{k}),\ 1\le k\le n\},\quad \alpha_{k}\in\RE\,,
\ee
and the ones in the $\beta$-family have self-adjointness domains
\be\label{beta}
\dom(\D_{\beta})=\{\Psi\equiv(\psi_{1},\psi_{2})\in H^{1}(\RE\backslash Y)\oplus 
H^{1}(\RE):[\psi_{1}]_{y_{k}}=-i\beta_{k}\psi_{2}(y_{k}),\ 1\le k\le n\},\quad \beta_{k}\in\RE\,.
\ee
Since the cases where all the $\alpha_{k}$'s or all the $\beta_{k}$'s are equal to zero correspond to $\D$, and the cases 
where there are $0<m<n$ $\alpha_{k}$'s or $\beta_{k}$'s which are zero reduce to the cases  with $(n-m)$ point interactions, without loss of generality we can suppose that all the $\alpha_{k}$'s and $\beta_{k}$'s
are different from zero. By Theorem \ref{K3} and Remark \ref{H1H2}, one has
$$
\D_{\alpha}=\D_{\Pi^{(\alpha)}, \Theta^{(\alpha)}}\,,\quad \Pi^{(\alpha)}=\Pi^{(\alpha)}_{1}\oplus\dots\oplus\Pi^{(\alpha)}_{n}\,,\quad
\Theta^{(\alpha)}=\Theta^{(\alpha)}_{1}\oplus\dots\oplus\Theta^{(\alpha)}_{n}\,,
$$
where
$$
\Pi^{(\alpha)}_{k}(\xi_{1},\xi_{2})=(\xi_{1},0)\,,\qquad \Theta_{k}^{(\alpha)}:\CO\to\CO\,,\quad \Theta_{k}^{(\alpha)}=\alpha_{k}^{-1}
$$
and
$$
\D_{\beta}=\D_{\Pi^{(\beta)}, \Theta^{(\beta)}}\,,\quad \Pi^{(\beta)}=\Pi^{(\beta)}_{1}\oplus\dots\oplus\Pi^{(\beta)}_{n}\,,\quad
\Theta^{(\beta)}=\Theta^{(\beta)}_{1}\oplus\dots\oplus\Theta^{(\beta)}_{n}\,,
$$
where
$$
\Pi^{(\beta)}_{k}(\xi_{1},\xi_{2})=(0,\xi_{2})\,,\qquad \Theta_{k}^{(\beta)}:\CO\to\CO\,,\quad \Theta_{k}^{(\beta)}=\beta_{k}^{-1}\,.
$$
Therefore,
$$
(\D_{\alpha})^{2}=\widehat\H_{\alpha}+\frac{\uno}4\,,\qquad
$$
where
$$
\widehat \H_{\alpha}=\widehat\H _{\widehat \Pi^{(\alpha)}, \widehat \Theta^{(\alpha)}}\,,
$$
$$\ran\big(\widehat \Pi^{(\alpha)}\big)=\ran\big(U^{*}\big(\widehat \Pi^{(\alpha)}_{1}\oplus\dots\oplus\widehat \Pi^{(\alpha)}_{n}\big)U\big)\,,\quad
\widehat \Theta^{(\alpha)}=U^{*}\big(\widehat \Theta^{(\alpha)}_{1}\oplus\dots\oplus\widehat \Theta^{(\alpha)}_{n}\big)U\,,
$$
$$
\ran(\widehat\Pi^{(\alpha)}_{k})=\ran(M_{2}(\Pi^{(\alpha)}_{k}\oplus\Pi^{(\alpha)}_{k}))=\CO\oplus\{0\}\oplus\{0\}\oplus\CO\equiv\CO^{2}\,,
$$
$$
\widehat\Theta_{k}^{(\alpha)}=M_{1}(\Theta^{(\alpha)}_{k}\oplus\Theta^{(\alpha)}_{k})M_{2}^{-1}:\CO^{2}\to\CO^{2}\,,\quad \widehat\Theta_{k}^{(\alpha)}
=\frac1{\alpha_{k}}\begin{bmatrix}0&-i\\i&-1
\end{bmatrix}
$$
and
$$(\D_{\beta})^{2}=\widehat\H_{\beta}+\frac{\uno}4\,,
$$
where
$$
\widehat \H_{\beta}=\widehat\H _{\widehat \Pi^{(\beta)}, \widehat \Theta^{(\beta)}}\,,
$$
$$
\ran\big(\widehat \Pi^{(\beta)}\big)=\ran\big(U^{*}\big(\widehat \Pi^{(\beta)}_{1}\oplus\dots\oplus\widehat \Pi^{(\beta)}_{n}\big)U\big)\,,\quad
\widehat \Theta^{(\beta)}=U^{*}\big(\widehat \Theta^{(\beta)}_{1}\oplus\dots\oplus\widehat \Theta^{(\beta)}_{n}\big)U\,,
$$
$$
\ran(\widehat\Pi^{(\beta)}_{k})=\ran(M_{2}(\Pi^{(\beta)}_{k}\oplus\Pi^{(\beta)}_{k}))
=\{0\}\oplus\CO\oplus\CO\oplus\{0\}\equiv\CO^{2}\,,
$$
$$
\widehat\Theta_{k}^{(\beta)}=M_{1}(\Theta^{(\beta)}_{k}\oplus\Theta^{(\beta)}_{k})M_{2}^{-1}:\CO^{2}\to\CO^{2}\,,\quad \widehat\Theta_{k}^{(\beta)}
=\frac1{\beta_{k}}\begin{bmatrix}1&i\\-i&0
\end{bmatrix}.
$$
Hence,
\begin{align*}
&\dom(\widehat\H_{\alpha})=\{\Psi\equiv(\psi_{1},\psi_{2})\in H^{2}(\RE\backslash Y):
[\psi_{1}]_{y_{k}}=[\psi'_{2}]_{y_{k}}=0,\ \\
&[\psi'_{1}]_{y_{k}}=\alpha_{k}(\psi_{1}(y_{k})-i\psi'_{2}(y_{k})),\ 
[\psi_{2}]_{y_{k}}=-i\alpha_{k}\psi_{1}(y_{k}),\ 1\le k\le n\}\,,
\end{align*}
and
\begin{align*}
&\dom(\widehat\H_{\beta})=\{\Psi\equiv(\psi_{1},\psi_{2})\in H^{2}(\RE\backslash Y):
[\psi'_{1}]_{y_{k}}=[\psi_{2}]_{y_{k}}=0,\ \\
&[\psi_{1}]_{y_{k}}=-i\beta_{k}\psi_{2}(y_{k}),\ 
[\psi'_{2}]_{y_{k}}=-\beta_{k}(\psi_{2}(y_{k})+i\psi'_{1}(y_{k})),\ 1\le k\le n\}\,.
\end{align*}
\subsection{Separating boundary conditions}\label{SBC}
Let $n=1$, and $\Pi=\uno$. By $2\tau\Psi=\Psi(y^{-})+\Psi(y^{+})$ and $\rho\Psi=i\sigma_{1}(\Psi(y_{+})-\Psi(y^{+}))$, the boundary condition $\tau\Psi=\Theta\rho\Psi$ 
rewrites as
$$
(2i\Theta\sigma_{1}+\uno)\Psi(y^{-})=(2i\Theta\sigma_{1}-\uno)\Psi(y^{+})\,.
$$ 
If $\Theta$ is such that
	\be\label{range}
	\ran(2i\Theta\sigma_{1}+\uno)\cap \ran(2i\Theta\sigma_{1}-\uno)=\{0\}\,,
	\ee
	then $\tau\Psi=\Theta\rho\Psi$ is equivalent to the separating boundary conditions
	\begin{align}
	&(2i\Theta\sigma_{1}+\uno)\Psi(y^{-})=0\label{sep-}
	\\&(2i\Theta\sigma_{1}-\uno)\Psi(y^{+})=0\label{sep+}\,.
	\end{align}
By the equivalence of \eqref{range} with
	\be\label{det}
	\det(2i\sigma_{1}\Theta-\uno)=0\,, 
	\ee
one gets that \eqref{range} holds if and only if
          \be\label{Tsep}
	\Theta=\Theta_{\omega,\alpha,\beta}:=
	\frac12\begin{bmatrix}\alpha&i\omega\sqrt{1+\alpha\beta}\ \\	-i\omega\sqrt{1+\alpha\beta}&\beta\end{bmatrix},\qquad
	\omega\in\{-1,+1\}\,,\ \alpha,\beta\in\RE\,,\ \alpha\beta\ge -1\,.
	\ee
For such a $\Theta$,  the boundary conditions \eqref{sep-}, \eqref{sep+} can be rewriten, whenever $\Psi\equiv(\psi_{1},\psi_{2})$, as
		
\be\label{eta-}
\psi_{2}(y^{-})=i\eta^{-}_{\omega,\alpha,\beta}\,\psi_{1}(y^{-})\,,		
\ee
\be\label{eta+}
\psi_{2}(y^{+})=i\eta^{+}_{\omega,\alpha,\beta}\,\psi_{1}(y^{+})\,,		
\ee
where
$$
\eta^{\pm}_{\omega,\alpha,\beta}:=\begin{cases}
-{\alpha^{-1}}(\omega\sqrt{1+\alpha\beta}\pm 1\,)\equiv{-\beta}{(\omega\sqrt{1+\alpha\beta}\mp 1\,)^{-1}}& \alpha\not=0\,,\ \omega\sqrt{1+\alpha\beta}\mp 1\not=0
\\
\mp\,2\alpha^{-1}&\alpha\not=0\,,\ \beta=0\,,\ \omega=\pm 1\\
\pm\,2^{-1}\beta&\alpha=0\,,\ \omega=\mp 1\\
\infty& \text{otherwise}
		\end{cases}
		$$
and the boundary condition $\psi_{2}(y^{\pm})=i\infty\psi_{1}(y^{\pm})$ is to be understood as $\psi_{1}(y^{\pm})=0$.\par
Then	
	$$
	\D_{\omega,\alpha,\beta}=\D^{-}_{\omega,\alpha,\beta}\oplus\D_{\omega,\alpha,\beta}^{+}\,,
	$$
	where $\D_{\omega,\alpha,\beta}:=\D_{\Theta_{\omega,\alpha,\beta}}$ and the self-adjoint operators $\D_{\omega,\alpha,\beta}^{-}$ and $\D^{+}_{\omega,\alpha,\beta}$ denote the Dirac operators in $L^{2}((-\infty,y);\CO^{2})$ and $L^{2}((y,+\infty);\CO^{2})$, with boundary conditions \eqref{sep-} and \eqref{sep+} (equivalently, \eqref{eta-} and \eqref{eta+}) respectively; let us remark that separating boundary conditions of the kind \eqref{sep-}, \eqref{sep+} (resp. \eqref{eta-}, \eqref{eta+}) 
	already appeared in \cite[Prop. 2.2]{HT} (resp. in \cite[Rem. 3.2]{BHT}).\par
Rewriting the boundary condition $\widehat\tau\Psi=\widehat\Theta_{\omega,\alpha,\beta}(\widehat\rho\uno)\Psi$ 
as
$$
\big(2i\widehat\Theta_{\omega,\alpha,\beta}(\sigma_{2}\oplus\sigma_{2})+\uno\big)\widehat\Psi(y^{-})=
\big(2i\widehat\Theta_{\omega,\alpha,\beta}(\sigma_{2}\oplus\sigma_{2})-\uno\big)\widehat\Psi(y^{+})\,,
$$ 	
where $\widehat\Psi\equiv(\psi_{1},\psi'_{1},\psi_{2},\psi'_{2})$ and $\widehat\Theta_{\omega,\alpha,\beta}$ is defined 
by \eqref{wT1}, i.e., 
$$
	\widehat\Theta_{\omega,\alpha,\beta} = \begin{bmatrix} 0 & \omega\sqrt{1+\alpha\beta} & 0 & -i\alpha \\ \omega\sqrt{1+\alpha\beta} & \beta & i\beta & 0 \\ 0 & -i\beta & 0 & -\omega\sqrt{1+\alpha\beta} \\i\alpha& 0 & -\omega\sqrt{1+\alpha\beta} & -\alpha\end{bmatrix}\,,
$$
one can check that $$
\det\big(2i(\sigma_{2}\oplus\sigma_{2})\widehat\Theta_{\omega,\alpha,\beta}-\uno\big)=0
$$
and so, proceeding as above, 
$$ 
\ran\big(2i\widehat\Theta_{\omega,\alpha,\beta}(\sigma_{2}\oplus\sigma_{2})+\uno\big)\cap
\ran\big(2i\widehat\Theta_{\omega,\alpha,\beta}(\sigma_{2}\oplus\sigma_{2})-\uno\big)=\{0\}\,.
$$
Thus the separating boundary conditions 
\begin{align}
&\big(2i\widehat\Theta_{\omega,\alpha,\beta}(\sigma_{2}\oplus\sigma_{2})+\uno\big)\widehat\Psi(y^{-})=0\label{S-sep-}\\
&\big(2i\widehat\Theta_{\omega,\alpha,\beta}(\sigma_{2}\oplus\sigma_{2})-\uno\big)\widehat\Psi(y^{+})=0\label{S-sep+}
\end{align}
hold for $\widehat \H_{\omega,\alpha,\beta}:=\widehat \H_{\widehat\Theta_{\omega,\alpha,\beta}}$ and 
$$
\widehat \H_{\omega,\alpha,\beta}=\widehat \H_{\omega,\alpha,\beta}^{-}\oplus\widehat \H_{\omega,\alpha,\beta}^{+}\,,
$$
where the self-adjoint operators $\widehat\H_{\omega,\alpha,\beta}^{-}$ and $\widehat\H_{\omega,\alpha,\beta}^{+}$ denote the Schr\"odinger operator $-\frac{d^{2}}{dx^{2}}\,\uno$ in $L^{2}((-\infty,y);\CO^{2})$ and $L^{2}((y,+\infty);\CO^{2})$, with boundary conditions \eqref{S-sep-} and \eqref{S-sep+} respectively.
Furthermore,
$$
(\D_{\omega,\alpha,\beta}^{\pm})^{2}=\widehat \H_{\omega,\alpha,\beta}^{\pm}+\frac{\uno}4\,.
$$
By \eqref{eta-}, \eqref{eta+} and by Remark \ref{equi}, the separating boundary conditions \eqref{S-sep-}, \eqref{S-sep+} for $\widehat \H_{\omega,\alpha,\beta}^{\pm}$ rewrite, whenever $\Psi\equiv(\psi_{1},\psi_{2})$, as
$$
\psi_{2}(y^{\pm})=i\eta^{\pm}_{\omega,\alpha,\beta}\,\psi_{1}(y^{\pm})\,,\qquad 		
i\eta^{\pm}_{\omega,\alpha,\beta}\,\psi'_{2}(y^{\pm})=\psi'_{1}(y^{\pm})+\eta^{\pm}_{\omega,\alpha,\beta}\,\psi_{1}(y^{\pm})\,.		
$$
In the case $n=1$, $\Pi\not=\uno$,   
the boundary conditions in $\dom(\D_{\Pi,\Theta})$
give 
\be\label{notloc}
(\Pi-\uno)\sigma_{1}\Psi(y^{-})=(\Pi-\uno)\sigma_{1}\Psi(y^{+})\,,\qquad
\Pi(2i\theta\sigma_{1}+\uno)\Psi(y^{-})=\Pi(2i\theta\sigma_{1}-\uno)\Psi(y^{+})\,.
\ee
Since $\ran((\Pi-\uno)\sigma_{1})=\ran(\Pi-\uno)$ and, by $\det(2i\theta\sigma_{1}\pm\uno)=1+4\theta^{2}\not=0$, 
$\ran(\Pi(2i\theta\sigma_{1}\pm\uno))=\ran(\Pi)$, the relations \eqref{notloc} do not allow any separating boundary conditions. \par
By the $n=1$ case, one immediately gets the family of separating and local boundary conditions: it suffices to take 
$\Theta_{\underline\omega,\underline\alpha,\underline\beta}:=\Theta_{\omega_{1},\alpha_{1},\beta_{1}}\oplus\dots\oplus\Theta_{\omega_{n},\alpha_{n},\beta_{n}}$. 
Then, using the abbreviated notations $\D_{\underline\omega,\underline\alpha,\underline\beta}\equiv
\D_{\Theta_{\underline\omega,\underline\alpha,\underline\beta}}$ and 
$\widehat\H_{\underline\omega,\underline\alpha,\underline\beta}\equiv
\widehat\H_{\widehat\Theta_{\underline\omega,\underline\alpha,\underline\beta}}$, where $\widehat\Theta_{\underline\omega,\underline\alpha,\underline\beta}:=U^{*}\big(\widehat\Theta_{\omega_{1},\alpha_{1},\beta_{1}}\oplus\dots\oplus\widehat\Theta_{\omega_{n},\alpha_{n},\beta_{n}}\big)U$ and $\widehat\Theta_{\omega_{k},\alpha_{k},\beta_{k}}$ is defined 
by \eqref{wT1}, i.e., $\widehat\Theta_{\omega_{k},\alpha_{k},\beta_{k}}:=M_{1}(\Theta_{\omega_{k},\alpha_{k},\beta_{k}}\oplus \Theta_{\omega,\alpha,\beta})M_{2}^{-1}$,
one obtains 
$$
\D_{{\underline\omega},\underline\alpha,\underline\beta}=\D^{-}_{\omega_{1},\alpha_{1},\beta_{1}}\oplus\D_{\omega_{1,2},\alpha_{1,2},\beta_{1,2}}\oplus \D_{\omega_{2,3},\alpha_{2,3},\beta_{2,3}}\oplus\dots\oplus\D_{\omega_{n-1,n},\alpha_{n-1,n},\beta_{n-1,n}}\oplus
\D^{+}_{\omega_{n},\alpha_{n},\beta_{n}}
$$
and
$$
\widehat\H_{{\underline\omega},\underline\alpha,\underline\beta}=\widehat\H^{-}_{\omega_{1},\alpha_{1},\beta_{1}}\oplus\widehat\H_{\omega_{1,2},\alpha_{1,2},\beta_{1,2}}\oplus \widehat\H_{\omega_{2,3},\alpha_{2,3},\beta_{2,3}}\oplus\dots\oplus\widehat\H_{\omega_{n-1,n},\alpha_{n-1,n},\beta_{n-1,n}}\oplus
\widehat\H^{+}_{\omega_{n},\alpha_{n},\beta_{n}}\,.
$$
Here $\D_{\omega_{k-1,k},\alpha_{k-1,k},\beta_{k-1,k}}$ denotes the self-adjoint Dirac operator in 
$L^{2}((y_{k-1}, y_{k});\CO^{2})$ with boundary conditions of the kind \eqref{sep+} at $y_{k-1}$ (with parameters $\omega_{k-1},\alpha_{k-1},\beta_{k-1}$) and of the kind \eqref{sep-} at $y_{k}$ (with parameters $\omega_{k},\alpha_{k},\beta_{k}$); $\widehat\H_{\omega_{k-1,k},\alpha_{k-1,k},\beta_{k-1,k}}$ is defined in a similar way, using the boundary conditions \eqref{S-sep-} and \eqref{S-sep+}. Furthermore, 
$$
(\D_{\omega_{k-1,k},\alpha_{k-1,k},\beta_{k-1,k}})^{2}=\widehat\H_{\omega_{k-1,k},\alpha_{k-1,k},\beta_{k-1,k}}+\frac{\uno}4\,,\qquad 1\le k\le n\,.
$$ 
	\subsection{Supersymmetry}\label{SUSY}
Since 
$$
\sigma_{1}\sigma_{2}+\sigma_{2}\sigma_{1}=0=\sigma_{3}\sigma_{2}+\sigma_{2}\sigma_{3}\,,
$$ 
one has
$$
\sigma_{2}\D_{\RE \backslash Y}+\D_{\RE \backslash Y}\sigma_{2}=0\,.
$$
Therefore, if $(\Pi,\Theta)$ is such that
\be\label{susy}
\begin{cases}\rho\Psi\in\ran(\Pi)\\
\Pi\tau\Psi=\Theta\rho\Psi
\end{cases}
\quad\Longrightarrow\quad
\begin{cases}\rho\sigma_{2}\Psi\in\ran(\Pi)\\
\Pi\tau\sigma_{2}\Psi=\Theta\rho\sigma_{2}\Psi\,,
\end{cases}
\ee
then $\sigma_{2}$ anti-commutes with $\D_{\Pi,\Theta}$ and so, by \eqref{D2=H}, the system 
\be\label{sup-syst}
\left(\widehat\H_{\widehat\Pi,\widehat\Theta}+\frac{\uno}4,\sigma_{2},\D_{\Pi,\Theta}\right)
\ee 
has supersymmetry (see, e.g., \cite[Chapter 1]{Arai}, \cite[Section 6.3]{Simon}).\par
By
$$
\langle\sigma_{2}\Psi\rangle_{y}=\sigma_{2}\langle\Psi\rangle_{y}\,,\qquad
[\sigma_{2}\Psi]_{y}=\sigma_{2}[\Psi]_{y}\,,
$$
and by $\sigma_{1}\sigma_{2}=-\sigma_{2}\sigma_{1}$,
one gets
$$
\tau\sigma_{2}\Psi=\sigma_{2}^{\oplus}\tau\Psi\,,\qquad\rho\sigma_{2}\Psi=-\sigma_{2}^{\oplus}\rho\Psi\,,\qquad
 \sigma_{2}^{\oplus}:=\sigma_{2}\oplus\dots\oplus\sigma_{2}\,.
$$
Therefore, \eqref{susy} holds whenever
\be\label{susy2}
\Pi\sigma^{\oplus}_{2}-\sigma^{\oplus}_{2}\Pi=0=\Theta\sigma^{\oplus}_{2}
+\sigma_{2}^{\oplus}\Theta\,.
\ee
Given a couple $(\Pi,\Theta)$ which satisfies \eqref{susy2}, let us further suppose that 
\be\label{susy-det}
\det(\Theta+\Pi\tau G_{0}\Pi)\not=0\,.
\ee
Then, by \eqref{detD}, zero is not an eigenvalue of $\D_{\Pi,\Theta}$, i.e., the system \eqref{sup-syst} has no supersymmetric state and there is a spontaneous supersymmetry breaking (see, e.g., \cite[Section 1.8]{Arai}).  
In the case $n=1$, the solutions of \eqref{susy2} are found immediately. \par If $\Pi\not=\uno$, then $\Pi=\Pi_{\pm}:=|v_{\pm}\rangle\langle v_{\pm}|$ and $\Theta=0$, where $v_{\pm}$, $|v_{\pm}|=1$, solves $\sigma_{2}v_{\pm}=\pm v_{\pm}$. \par If $\Pi=\uno$, then $\Theta=\Theta_{a,b}:=b\sigma_{1}+a\sigma_{3}$, $a,b\in\RE$. \par
Since, by \eqref{Djk}, $\tau G_{0}=-\frac12\,\sigma_{3}$, $\det(\Theta_{a,b}+\tau G_{0})=0$ if and only if $b=0$ and $a=\frac12$. Therefore, for any $(a,b)\in\RE^{2}\backslash \{(\frac12,0)\}$ the system \eqref{sup-syst} with $\Pi=\uno$ and $\Theta=\Theta_{a,b}$ has no supersymmetric state and there is a spontaneous supersymmetry breaking. 
\par
Notice that once the solution of \eqref{susy2} in known in the $n=1$ case, then the set of solutions for the case of $n>1$ local point interactions is readily obtained: $\Pi=\Pi_{1}\oplus\dots\oplus\Pi_{n}$ and $\Theta=\Theta_{1}\oplus\dots\oplus\Theta_{n}$, where $(\Pi_{k},\Theta_{k})$ is equal either to $(\Pi_{\pm},\zero)$ or to $(\uno,\Theta_{a_{k},b_{k}})$.

\subsection{Quantum Graphs.}\label{QG} Since $\D_{\Pi,\Theta}$ is a generic self-adjoint extension of $S=\D|C^{\infty}_{comp}(\RE\backslash Y;\CO^{2})=(\D|C^{\infty}_{comp}(I_{0};\CO^{2}))\oplus\cdots\oplus (\D|C^{\infty}_{comp}(I_{n};\CO^{2}))$, 
the  nonlocal extensions of $S$ provide  the self-adjoint realizations of the Dirac operator on a quantum graph with the two ends $\overline I_{0}=(-\infty,y_{1}]$ and $\overline I_{n}=[y_{n},+\infty)$ and the $(n-1)$ edges $\overline I_{1}=[y_{1},y_{2}],\dots,\overline I_{n-1}=[y_{n-1},y_{n}] $; the boundary conditions corresponding to the couple $(\Pi,\Theta)$ specify the connectivity of the graph. The case of a compact graph can be obtained by imposing separating boundary conditions at the two ends. Likewise,   the  nonlocal extensions of $\H\uno|C^{\infty}_{comp}(\RE\backslash Y;\CO^{2})$ provide  self-adjoint realizations of the Schr\"odinger operator on a quantum graph with two ends and $(n-1)$ edges. For an introduction to the theory of quantum graphs we refer to the book \cite{BK} and the many references there; however, let us point out that our way of building the self-adjoint realizations on the graph is not the standard one.  
\par
As an explicit example, let us consider the Dirac operator on the eye graph (see \cite[Section III.D]{bulla}). Therefore, we choose the subclass of boundary conditions for the Dirac operator in $L^{2}(G;\CO^{2})$, $G=(-\infty, y_{1}]\sqcup[y_{1},y_{2}]\sqcup[y_{2},y_{3}]\sqcup[y_{3},+\infty)$, 
connecting $\Psi(y_{1}^{-})$ with both $\Psi(y_{1}^{+})$ and $\Psi(y_{2}^{+})$ and 
connecting $\Psi(y_{3}^{+})$ with both $\Psi(y_{2}^{-})$ and $\Psi(y_{3}^{-})$. Such kind of boundary conditions give to $G$ the topology of a circle with two ends. 
$$
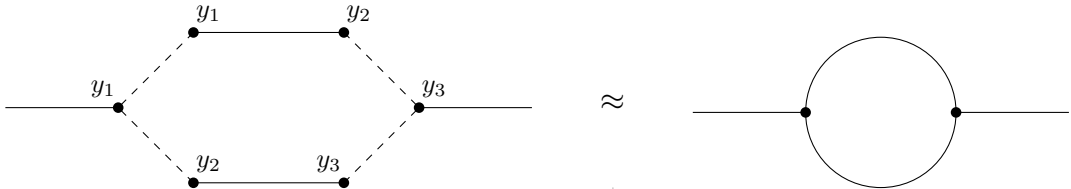
\begin{tikzpicture}[main/.style = {draw, circle}] 
		\fill   (0,0) circle[radius=2pt] node [above,font=\small ,xshift=-2mm] {$y_1$};
		\fill   (1,1) circle[radius=2pt] node [above,font=\small ,xshift=2mm] {$y_1$};
		\fill   (3,1) circle[radius=2pt] node [above,font=\small ,xshift=2mm] {$y_2$};
		\fill   (4,0) circle[radius=2pt] node [above,font=\small ,xshift=2mm] {$y_3$};
		\fill   (3,-1) circle[radius=2pt] node [above,font=\small ,xshift=-2mm] {$y_3$};
		\fill   (1,-1) circle[radius=2pt] node [above,font=\small ,xshift=2mm] {$y_2$};
		\draw (0,0) -- +(-1.5,0);
		\draw (1,1) -- (3,1);
		\draw (1,-1) -- (3,-1);
		\draw (4,0) -- +(1.5,0);
		\draw[dashed]  (1,1) -- (0,0);
		\draw[dashed]  (3,1) -- (4,0);
		\draw[dashed]  (3,-1) -- (4,0);
		\draw[dashed]  (1,-1) -- (0,0);  
\end{tikzpicture}
\qquad 
\begin{tikzpicture}
	\fill   (0,0) circle[radius=0pt] node [above,font=\large ,yshift=9mm] {$\approx$};
\end{tikzpicture}
\qquad 
\begin{tikzpicture}
	\draw   (0,0) circle[radius=10mm] ;
	\fill   (1,0) circle[radius=2pt];
	\fill   (-1,0) circle[radius=2pt];
	\draw (1,0) node[circle]{} -- +(1.5,0);
	\draw (-1,0) node[circle]{} -- +(-1.5,0);
\end{tikzpicture}
$$
Furthermore, we restrict to Kirchhoff-type boundary conditions, meaning that we select the ones which, in the non relativistic limit, correspond to Kirchhoff (or free) boundary conditions for the Schr\"odinger operator in $L^{2}(G)$ (see \cite{BCT1}, \cite{BCT2}, \cite{BCT3}). Therefore, we require, for $\Psi\equiv(\psi_{1},\psi_{2})$ in the self-adjointness domain,
\be\label{Kirch}
\begin{cases}
\psi_{1}(y_{1}^{-})=\psi_{1}(y_{1}^{+})=\psi_{1}(y_{2}^{+})&\\ 
\psi_{2}(y_{1}^{-})-\psi_{2}(y_{1}^{+})-\psi_{2}(y_{2}^{+})=0&\\
\psi_{1}(y_{2}^{-})=\psi_{1}(y_{3}^{-})=\psi_{1}(y_{3}^{+})&\\
\psi_{2}(y_{2}^{-})+\psi_{2}(y_{3}^{-})-\psi_{2}(y_{3}^{+})=0\,.&
\end{cases}
\ee
These boundary conditions rewrite as 
\be\label{Gbc}
\begin{cases}
[\psi_{1}]_{y_{1}}=[\psi_{1}]_{y_{3}}=0&\\
[\psi_{2}]_{y_{1}}=-\psi_{2}(y_{2}^{+})&\\
[\psi_{2}]_{y_{3}}=\psi_{2}(y_{2}^{-})&\\
\langle\psi_{1}\rangle_{y_{1}}=\psi_{1}(y_{2}^{+})&\\
\langle\psi_{1}\rangle_{y_{3}}=\psi_{1}(y_{2}^{-})&
\end{cases}
\equiv\quad
\begin{cases}
[\psi_{1}]_{y_{1}}=[\psi_{1}]_{y_{3}}=0&\\
[\psi_{2}]_{y_{1}}+[\psi_{2}]_{y_{2}}+[\psi_{2}]_{y_{3}}=0&\\
\langle\psi_{1}\rangle_{y_{1}}-\langle\psi_{1}\rangle_{y_{3}}=[\psi_{1}]_{y_{2}}&\\
\langle\psi_{1}\rangle_{y_{1}}-2\langle\psi_{1}\rangle_{y_{2}}+\langle\psi_{1}\rangle_{y_{3}}=0&\\
2\langle\psi_{2}\rangle_{y_{2}}=[\psi_{2}]_{y_{3}}-[\psi_{2}]_{y_{1}}\,.&
\end{cases}
\ee
The relations $[\psi_{1}]_{y_{1}}=[\psi_{1}]_{y_{3}}=[\psi_{2}]_{y_{1}}+[\psi_{2}]_{y_{2}}+[\psi_{2}]_{y_{3}}=0$ in \eqref{Gbc} coincide with $\rho\Psi\in\ran(\Pi)$, where the orthogonal projector $\Pi:\CO^{6}\to\CO^{6}$ is represented by the matrix
\be\label{PG}
\Pi=
\frac13
\begin{bmatrix}
2&0&-1&0&-1&0\\
0&0&0&0&0&0\\
-1&0&2&0&-1&0\\
0&0&0&3&0&0\\
-1&0&-1&0&2&0\\
0&0&0&0&0&0
\end{bmatrix}.
\ee
Then, one can easily check that the other relations in \eqref{Gbc} are equivalent to $\Pi\tau\Psi=\Theta\rho\Psi$ whenever $\Theta$, a symmetric operator in the $3$-dimensional subspace $\ran(\Pi)$, is represented, as a symmetric linear operator in $\CO^{6}$ preserving $\ran(\Pi)$, by the Hermitian  matrix
\be\label{TG}
\Theta=\frac{i}2\begin{bmatrix}
0&0&0&-1&0&0\\
0&0&0&0&0&0\\
0&0&0&0&0&0\\
1&0&0&0&-1&0\\
0&0&0&1&0&0\\
0&0&0&0&0&0
\end{bmatrix}.
\ee
Therefore, using $L^{2}(G;\CO^{2})\equiv L^{2}(\RE;\CO^{2})$, one gets $\D_{\rm K}\equiv \D_{\Pi,\Theta}$, where $\Pi$ and $\Theta$ are as in \eqref{PG} and \eqref{TG} respectively and where $\D_{\rm K}$ denotes the Dirac operator on the eye graph with the Kirchhoff boundary conditions at the two vertices.
Then,  by  Remark \ref{equi2} and by \eqref{Kirch},  the Schr\"odinger operator $\widehat\H_{\widehat\Pi,\widehat\Theta}$ satisfies both the boundary conditions ${\rm K}$ and ${\rm K}_{*}$, where 
\be\label{bcGS}
{\rm K}\equiv\begin{cases}
\psi_{1}(y_{1}^{-})=\psi_{1}(y_{1}^{+})=\psi_{1}(y_{2}^{+})&\\ 
\psi_{1}'(y_{1}^{-})-\psi_{1}'(y_{1}^{+})-\psi_{1}'(y_{2}^{+})=0&\\
\psi_{1}(y_{2}^{-})=\psi_{1}(y_{3}^{-})=\psi_{1}(y_{3}^{+})&\\
\psi_{1}'(y_{2}^{-})+\psi_{1}'(y_{3}^{-})-\psi_{1}'(y_{3}^{+})=0&
\end{cases}
\quad
{\rm K}_{*}\equiv\begin{cases}
\psi_{2}'(y_{1}^{-})=\psi_{2}'(y_{1}^{+})=\psi_{2}'(y_{2}^{+})&\\ 
\psi_{2}(y_{1}^{-})-\psi_{2}(y_{1}^{+})-\psi_{2}(y_{2}^{+})=0&\\
\psi_{2}'(y_{2}^{-})=\psi_{2}'(y_{3}^{-})=\psi_{2}'(y_{3}^{+})&\\
\psi_{2}(y_{2}^{-})+\psi_{2}(y_{3}^{-})-\psi_{2}(y_{3}^{+})=0\,.&
\end{cases}
\ee
This gives 
\be\label{DHG}
(\D_{\rm K})^{2}=\left(\H_{\rm K}+\frac14\right)\oplus\left(\H_{{\rm K}_{*}}+\frac14\right)\,,
\ee
where $\H_{\rm K}$ is the Schr\"odinger  operator in $L^{2}(G)$ with the boundary conditions ${\rm K}$ and $\H_{{\rm K}_{*}}$ is the Schr\"odinger  operator in $L^{2}(G)$ with the  boundary conditions ${\rm K}_{*}$. The boundary conditions ${\rm K}$ coincide with the usual Kirchhoff ones (see \cite[eq. (1.4.4)]{BK}) while the boundary conditions ${\rm K}_{*}$, are as sort of reversed Kirchhoff ones (named ''homogeneous $\delta'$ vertex conditions'' in \cite{BCT2}) given by the exchange $\psi\leftrightarrow\psi'$. The boundary conditions ${\rm K_{*}}$, like the ${\rm K}$ ones, give, in the case of the real line, the free Schr\"odinger operator; thus, \eqref{DHG} is consistent with \eqref{Rel}. Furthermore, the Schr\"odinger operator $\H_{\rm K}\oplus\H_{{\rm K}_{*}}$ appears in the nonrelativistic limit of $\D_{K}$, see \cite[Proposition 1.3]{BCT2}.\par
The arguments in the previous example extend to any graph: by Remark \ref{equi2}, to the Kirchhoff-type boundary conditions for the Dirac operator $\D_{\rm K}$ on the graph, i.e.,  to 
$$
\begin{cases}
\text{$\psi_{1}$ continuous at any vertex $v$}&\\
\sum_{v}^{\pm}\psi_{2}(v)=0\ \text{for any vertex $v$,}&
\end{cases}
$$
correspond, for the  Schr\"odinger operators $\H_{\rm K}$ and $\H_{{\rm K}_{*}}$ such that \eqref{DHG} holds,  the 
 boundary conditions
$$
{\rm K}\equiv\begin{cases}
\text{$\psi$ continuous at any vertex $v$}&\\
\sum_{v}^{\pm}\psi'(v)=0\,\ \text{for any vertex $v$}&
\end{cases}
\quad
{\rm K}_{*}\equiv\begin{cases}
\text{$\psi'$ continuous an any vertex $v$}&\\
\sum_{v}^{\pm}\psi(v)=0\,\ \text{for any vertex $v$}.&
\end{cases}
$$
Here, $\sum_{v}^{\pm}f(v)$ means the sum over all the points $y_{k}\in Y$ corresponding to the vertex $v$ with the sign convention 
$$
f(v):=\begin{cases}
-f(y_{k}^{+})&\text{$y_{k}$ is at the left end of the interval/half-line}\\
+f(y_{k}^{-})&\text{$y_{k}$ is at the right end of the interval/half-line.}
\end{cases}
$$ 
\vskip 10pt\noindent
{\bf Acknowledgements.} We thank an anonymous referee for the accurate reading and for the useful remarks.


\begin{thebibliography}{99}


\bibitem{AGHKH} S. Albeverio, F. Gesztesy, R. Hoegh-Krohn, H. Holden:
{\it Solvable Models in Quantum Mechanics}, 2nd Edition, AMS Chelsea Publ., 2005.

\bibitem{AKK} S. Albeverio, W. Karwowski, V. Koshmanenko: Square powers of singularly perturbed operators. {\it Math. Nachr.} {\bf 173} (1995), 5-24.

\bibitem{ASSW} M.H. Al-Hashimi, M. Salman, A. Shalaby, U.-J. Wiese: Supersymmetric descendants of self-adjointly extended quantum mechanical Hamiltonians. {\it Ann. Physics} {\bf 337} (2013), 1-24. 

\bibitem{Arai} A. Arai: {\it Infinite-dimensional Dirac Operators and Supersymmetric Quantum Fields.} Springer, 2022.
 
\bibitem {ArMaVe 1}N. Arrizabalaga, A. Mas, L. Vega: Shell interactions for
Dirac operators. \emph{J. Math. Pures Appl.} (9) \textbf{102}(4), 617-639, 2014.

\bibitem {BEHL19}J. Behrndt, P. Exner, M. Holzmann, V. Lotoreichik: On Dirac operators in $\mathbb{R}^3$ with electrostatic and 
Lorentz scalar $\delta$-shell interactions. \emph{Quantum Stud. Math. Found.}, \textbf{6}, 295-314, 2019.

\bibitem{BHT} J. Behrndt, M. Holzmann, M. Tu\v sek: Two-dimensional Dirac operators with general $\delta$-shell interactions supported on a straight line. {\it J. Phys. A: Math. Theor.} {\bf 56}, 045201, 2023.

\bibitem{BLL} J. Behrndt, M. Langer, V. Lotoreichik:
Schr\"odinger operators with $\delta$ and $\delta'$-potentials supported on hypersurfaces. {\it 
Ann. Henri Poincar\'e} {\bf 14} (2013), 385-423.

\bibitem{BD} S. Benvegn\`u, L. D\c{a}browski: Relativistic point interaction,
\emph{Lett. in Math. Physics}, \textbf{30} (1994), 159-167.

\bibitem{BK} G. Berkolaiko, P. Kuchment: {\it Introduction to Quantum Graphs.}
American Mathematical Society, 2013. 

\bibitem{BCT1} W. Borrelli, R. Carlone, L. Tentarelli: Nonlinear Dirac equation on graphs with localized nonlinearities: bound states and nonrelativistic limit. {\it SIAM J. Math. Anal.} {\bf 51} (2019), 
1046-1081.

\bibitem{BCT2} W. Borrelli, R. Carlone, L. Tentarelli: On the nonlinear Dirac equation on noncompact metric graphs. {\it J. Differential Equations} {\bf 278} (2021), 326-357.  

\bibitem{BCT3} W. Borrelli, R. Carlone, L. Tentarelli: A note on the Dirac operator with Kirchhoff-type vertex conditions on metric graphs. In A. Michelangeli (ed.): {\it Mathematical challenges of zero-range physics - models, methods, rigorous results, open problems}, 81-104, Springer, 2021.

\bibitem{bulla} W. Bulla, T. Trenkler: The free Dirac operator on compact and noncompact graphs.
{\it J. Math. Phys.} {\bf 31} (1990), 1157-1163.

\bibitem{CMP} R. Carlone, M. Malamud, A. Posilicano: 
On the spectral theory of Gesztesy-\v{S}eba realizations of 1-D Dirac operators with point interactions on a discrete set. {\it J. Differential Equations} {\bf 254} (2013), 3835-3902.
 
\bibitem{CLMT} B. Cassano, V. Lotoreichik, A. Mas, M. Tu\v sek: General $\delta$-shell interactions for the two-dimensional Dirac operator: self-adjointness and approximation. {\it Rev. Mat. Iberoam.} 2022, doi 10.4171/RMI/1354. 
 
\bibitem{Simon} H. L. Cycon, R. G.  Froese, W. Kirsch, B. Simon:
{\it Schr\"odinger operators with application to quantum mechanics and global geometry.}
Springer-Verlag, 1987.

 \bibitem{DG} L. D\c{a}browski, H. Grosse: On nonlocal point interactions in one, two, and three dimensions. {\it J. Math. Phys.} {\bf 26} (1985), 2777-2780. 

\bibitem{GS} F. Gesztesy, P. \v{S}eba: New analytically solvable models of relativistic point interactions, {\it Lett. Math. Phys.} {\bf 13} (1987), 345-358.

\bibitem{HT} L. Heriban, M. Tu\v sek: Non-self-adjoint relativistic point interaction in one dimension, {\it J. Math. Anal. Appl.} {\bf 516} (2022), 126536

\bibitem{KM} A.S. Kostenko, M.M. Malamud: 1-D Schr\"odinger operators with local point interactions on a discrete set, {\it J. Differential
Equations} {\bf 249} (2010), 253-304.

\bibitem{JDE16}A. Mantile, A. Posilicano, M. Sini: Self-adjoint elliptic
operators with boundary conditions on not closed hypersurfaces. \emph{J.
Differential Equations} \textbf{261} (2016) , 1-55.

\bibitem{Pank} K. Pankrashkin: Resolvents of self-adjoint extensions with mixed boundary conditions. {\it Rep. Math. Phys.} {\bf 58} (2006), 207-221.
 
\bibitem {JFA} A. Posilicano: A Kre\u{\i}n-like formula for singular
perturbations of self-adjoint operators and applications. \emph{J. Funct.
Anal.}, \textbf{183} (2001), 109-147.

\bibitem {O&M}A. Posilicano: Self-adjoint extensions of restrictions,
\emph{Oper. Matrices}, \textbf{2} (2008), 483-506.

\bibitem{UT} T. Uchino, I. Tsutsui: Supersymmetric quantum mechanics under point singularities. {\it J. Phys. A} {\bf 36} (2003),, 6821-6846.
 
\end{thebibliography}
\end{document}